\documentclass[submission,copyright,creativecommons]{eptcs}

\providecommand{\event}{GCM 2020}
\usepackage[utf8]{inputenc}% for umlauts and other special characters: new LaTeX version
\usepackage{graphicx}
\usepackage{cite}
\usepackage{amsmath}
\usepackage{amssymb}
\usepackage{amsfonts}
\usepackage{extarrows}
\usepackage{amstext}
\usepackage{stmaryrd}% required for \llbracket, \rrbrackets

\usepackage{latexsym}% required for \nexists
\usepackage{rotating}
\usepackage{color}% required for \textcolor{}{}ss
\usepackage{fancybox}%ovaler Rand                
\usepackage{tikz}
\usetikzlibrary{chains,fit,positioning,shapes.geometric}
\usepackage{url}
\usepackage{dsfont} % for math symbols, macro \N
\usepackage{cutwin}
\usepackage{algorithm2e}

\usepackage{tikz-atposition}% \atposition and \atpositionsmall macros
\usepackage{tikz-cellular}% for the cellular example
\usepackage{tikz-circlesnarrows1}% \Circle,\Arrow, Loop, as well as \node,\outgoing,\leep
\usepackage{tikz-colorsnstyles}
\usepackage{tikz-extension}% some extensions
\usepackage{tikz-derivation}% for derivations
\usepackage{tikz-vcond} % macros for HR conditions and hyperesdges
\usepackage{wrapfig}

\providecommand\ignore[1]{}% zum Auskommentierenss

\usepackage{mytheorem}
\usepackage{abbr-habel} 
\usepackage{abbr-habel-tikz}% abbreviations
\usepackage{abbr-habel-tikz2} % absbreviations
\usepackage{abbr-poskitt-plump}
\usepackage{abbr-sandmann-tikz}
\usetikzlibrary{positioning,decorations.pathreplacing,arrows.meta} %ohne snakes?
\usetikzlibrary{arrows,automata,calc}
\usepackage{tabularx}
\tikzstyle{round corners}=[rounded corners=0.5ex]
\usepackage{float}
\usepackage{graphics}
\usepackage{booktabs}

\providecommand{\napx}[1]{#1}%nonappendix version
\providecommand{\apx}[1]{}%appendix version
\providecommand{\shortv}[1]{}%short version
\providecommand{\longv}[1]{#1} %long version
\providecommand{\ap}[1]{} %application condition
\providecommand{\forget}[1]{} %simplification

\ignore{
\usepackage{textcomp}
\usepackage[ngerman,english]{babel}

\usepackage{lmodern}
\usepackage{mathrsfs}

\usepackage[utf8]{inputenc}
\usepackage[T1]{fontenc}

\usepackage{color}

\usepackage{amsmath}
\usepackage{amsfonts}
\usepackage{amssymb}
\usepackage{wasysym}
\usepackage{pgfplots}
\usetikzlibrary{shapes,arrows,positioning,calc,decorations.pathreplacing, decorations.markings }
\usepackage{graphicx}
\usepackage{pgf}
\pgfplotsset{compat=1.15}
\usepackage{tikz}
\usetikzlibrary{arrows,automata}
\usetikzlibrary{positioning}

\usepackage{tikz-atposition}

\usepackage{tikz-extension}
\usepackage{abbr-letters}

\usepackage{abbr-poskitt-plump}
\usepackage{abbr-radke}

\usepackage{abbr-habel}
\usepackage{abbr-habel-tikz}
\usepackage{abbr-habel-tikz2}
\usepackage{abbr-sandmann-tikz}
\usepackage{framed}
}

\renewcommand{\bigvee}{\vee}
\renewcommand{\epsilon}{\varepsilon}
\renewcommand{\phi}{\varphi}

\renewcommand{\select}{\mathrm{sel}}

\newcommand{\labb}{\lambda}

\newcommand{\ssys}{\textbf{\text{S}}}
\newcommand{\senv}{\textbf{\text{E}}}
\newcommand{\sys}{s}
\newcommand{\env}{e}
\newcommand{\Sys}{\mathcal{S}}
\newcommand{\SysEnv}{\Sys \Env}
\newcommand{\Env}{\mathcal{E}}

\newcommand{\ntop}{\raisebox{-0.7mm}{\circlearound{$\top$}}}
\newcommand{\on}{\raisebox{-0.7mm}{\circlearound{$\senv$}}}
\newcommand{\off}{\raisebox{-0.7mm}{\circlearound{$\ssys$}}}
\newcommand{\Pot}{\mathfrak{P}}

\newcommand{\LTL}{\text{LTL}}
\newcommand{\CTL}{\text{CTL}}
\newcommand{\LTLt}{}
\newcommand{\opX}{\textbf{\text{X}}}
\newcommand{\opW}{\textbf{\text{W}}}
\newcommand{\opU}{\textbf{\text{U}}}
\newcommand{\opG}{\textbf{\text{G}}}

\newcommand{\m}{\text{m}}
\newcommand{\prm}{\text{prm}}
\newcommand{\opA}{\textbf{\text{A}}}
\newcommand{\opE}{\textbf{\text{E}}}
\newcommand{\opAX}{\textbf{\text{AX}}}
\newcommand{\opEX}{\textbf{\text{EX}}}
\newcommand{\Op}{\textbf{\textit{O}}}

\newcommand{\opEF}{\textbf{\text{EF}}}
\newcommand{\opAG}{\textbf{\text{AG}}}
\newcommand{\opEG}{\textbf{\text{EG}}}
\newcommand{\opAU}{\textbf{\text{AU}}}
\newcommand{\opEU}{\textbf{\text{EU}}}
\newcommand{\opAW}{\textbf{\text{AW}}}
\newcommand{\opEW}{\textbf{\text{EW}}}
\newcommand{\FinSeq}{\text{FinSeq}(G, \R)}

\newcommand{\app}{\textnormal{app}}

\newcommand{\pred}{\textnormal{pred}}
\newcommand{\suc}{\textnormal{succ}}

\renewcommand{\ext}{\text{ext}}

\newcommand{\st}{*}

\renewcommand{\prule}{r}

\renewcommand{\emptyset}{\varnothing}
\ignore{
\newcommand{\PCS}{\textnormal{LTL}_{\textnormal{CS}}}
\newcommand{\GR}{\textnormal{LTL}_{\textnormal{lm}}}

\newcommand{\kSC}{\textnormal{LTL}_{k\textnormal{SC}}}
}

\newcommand{\PCS}{\sigma}
\newcommand{\GR}{\psi}

\newcommand{\kSC}{\psi_k}
\newcommand{\PCSctl}{\xi}

\newcommand{\kWC}{\zeta_k}

\newcommand{\SysEnvC}{\overline{\mathcal{SE}}}
\definecolor{orange-red}{rgb}{1.9, 0.5, 0.3}
\definecolor{whitegray}{rgb}{0.9, 0.9, 0.9}
\definecolor{lightergray}{rgb}{0.82, 0.82, 0.82}

\definecolor{new}{rgb}{0, 0, 0}
\tikzstyle{norm} = [thick]
\tikzstyle{imp}=[>=stealth,double distance=2pt, double,->]
\tikzstyle{track} = [semithick,
   double distance=1.4pt,
   preaction = {decorate},
   postaction = {draw,line width=1.4pt, lightgray}]
\tikzstyle{blocked} = [semithick,
   double distance=1.4pt,
   preaction = {decorate},
   postaction = {draw,line width=1.4pt, lightgray}]
\tikzstyle{innerWhite} = [semithick, white,line width=1.4pt, shorten >= 4.5pt] 
\tikzstyle{flow} = [thick,
   double distance=1.4pt,
   preaction = {decorate},
   postaction = {draw,line width=1.4pt, yellow}]
\tikzstyle{innerWhite} = [semithick, white,line width=1.4pt, shorten >= 4.5pt] 

\newcommand{\track}{
\raisebox{-2mm}{
\begin{tikzpicture}
\node[shape=circle,draw,inner sep=2pt] (1) at (0,0) {};
\node[shape=circle,draw,inner sep=2pt] (2) at (1,0) {};
\draw[track] (1) to (2); 
\node[below] at (0,-0.1) {\scriptsize \textbf{1}};
\node[below] at (1,-0.1) {\scriptsize \textbf{2}};
\end{tikzpicture}}
}

\newcommand{\trackO}{
\begin{tikzpicture}
\node[shape=circle,draw,inner sep=2pt] (1) at (0,0) {};
\node[shape=circle,draw,inner sep=2pt] (2) at (1,0) {};
\draw[track] (1) to (2); 

\end{tikzpicture}
}

\newcommand{\onecar}{
\begin{tikzpicture}
\node[shape=circle,draw,inner sep=2pt] (1) at (0,0) {};
\node[shape=circle,draw,inner sep=2pt] (2) at (1,0) {};

\node[
  draw,
  align=center,
  regular polygon,
  regular polygon sides=4,
  inner sep=-1pt, fill=white] (car) at (0.5,0.5) {\faAutomobile};
\draw[track] (1) to (2); 
\draw[thick, bend left] (1) to (car);
\draw[thick, bend left] (car) to (2);

\node[below] at (0,-0.1) {\scriptsize \textbf{1}};
\node[below] at (1,-0.1) {\scriptsize \textbf{2}};
\end{tikzpicture}}

\newcommand{\twocar}{
\begin{tikzpicture}
\node[shape=circle,draw,inner sep=2pt] (1) at (0,0) {};
\node[shape=circle,draw,inner sep=2pt] (2) at (1,0) {};

\node[
  draw,
  align=center,
  regular polygon,
  regular polygon sides=4,
  inner sep=-1pt, fill=white] (car1) at (0.5,0.5) {\faAutomobile};
\node[
  draw,
  align=center,
  regular polygon,
  regular polygon sides=4,
  inner sep=-1pt, fill=white] (car2) at (0.5,-0.5) {\faAutomobile};

\draw[track] (1) to (2);

\draw[thick, bend left] (1) to (car1);
\draw[thick, bend left] (car1) to (2);
\draw[thick, bend right] (1) to (car2);
\draw[thick, bend right] (car2) to (2);

\node[below] at (-0.06,-0.1) {\scriptsize \textbf{1}};
\node[below] at (1.06,-0.1) {\scriptsize \textbf{2}};

 \end{tikzpicture}}

\newcommand{\twocarNoIndex}{
\begin{tikzpicture}
\node[shape=circle,draw,inner sep=2pt] (1) at (0,0) {};
\node[shape=circle,draw,inner sep=2pt] (2) at (1,0) {};

\node[
  draw,
  align=center,
  regular polygon,
  regular polygon sides=4,
  inner sep=-1pt, fill=white] (car1) at (0.5,0.5) {\faAutomobile};
\node[
  draw,
  align=center,
  regular polygon,
  regular polygon sides=4,
  inner sep=-1pt, fill=white] (car2) at (0.5,-0.5) {\faAutomobile};

\draw[track] (1) to (2);

\draw[thick, bend left] (1) to (car1);
\draw[thick, bend left] (car1) to (2);
\draw[thick, bend right] (1) to (car2);
\draw[thick, bend right] (car2) to (2);

 \end{tikzpicture}}

\newcommand{\moveL}{
\begin{tikzpicture}
\node[shape=circle,draw,inner sep=2pt] (1) at (0,0) {};
\node[shape=circle,draw,inner sep=2pt] (2) at (1,0) {};
\node[shape=circle,draw,inner sep=2pt] (3) at (2,0) {};

\node[
  draw,
  align=center,
  regular polygon,
  regular polygon sides=4,
  inner sep=-1pt, fill=white] (car) at (0.5,0.5) {\faAutomobile};

\node[below] at (0,-0.1) {\scriptsize \textbf{1}};
\node[below] at (1,-0.1) {\scriptsize \textbf{2}};
\node[below] at (2,-0.1) {\scriptsize \textbf{3}};

\draw[thick, bend left] (1) to (car);
\draw[thick, bend left] (car) to (2);

\draw[track] (1) to (2);
\draw[track] (2) to (3); \end{tikzpicture}}

\newcommand{\moveR}{
\begin{tikzpicture}
\node[shape=circle,draw,inner sep=2pt] (1) at (0,0) {};
\node[shape=circle,draw,inner sep=2pt] (2) at (1,0) {};
\node[shape=circle,draw,inner sep=2pt] (3) at (2,0) {};

\node[below] at (0,-0.1) {\scriptsize \textbf{1}};
\node[below] at (1,-0.1) {\scriptsize \textbf{2}};
\node[below] at (2,-0.1) {\scriptsize \textbf{3}};

\node[
  draw,
  align=center,
  regular polygon,
  regular polygon sides=4,
  inner sep=-1pt, fill=white] (car) at (1.5,0.5) {\faAutomobile};

\draw[thick, bend left] (2) to (car);
\draw[thick, bend left] (car) to (3);

\draw[track] (1) to (2);
\draw[track] (2) to (3); \end{tikzpicture}}

\newcommand{\blocked}{
\raisebox{-1mm}{
\begin{tikzpicture}
\node[shape=circle,draw,inner sep=2pt] (1) at (0,0) {};
\node[shape=circle,draw,inner sep=2pt] (2) at (1,0) {};
\draw[blocked] (1) to (2); 
\node[
  draw,semithick,
  align=center,
  regular polygon,
  regular polygon sides=3,
  inner sep=-2.3pt, fill=orange-red] at (0.5,-0.02) 
{\raisebox{1.1mm}{\textbf{$\lightning$}}};

\node[below] at (0,-0.1) {\scriptsize \textbf{1}};
\node[below] at (1,-0.1) {\scriptsize \textbf{2}};
\end{tikzpicture}}
}

\newcommand{\blockedO}{
\begin{tikzpicture}
\node[shape=circle,draw,inner sep=2pt] (1) at (0,0) {};
\node[shape=circle,draw,inner sep=2pt] (2) at (1,0) {};
\draw[blocked] (1) to (2); 
\node[
  draw,semithick,
  align=center,
  regular polygon,
  regular polygon sides=3,
  inner sep=-2.3pt, fill=orange-red] at (0.5,-0.02) 
{\raisebox{1.1mm}{\textbf{$\lightning$}}};

\end{tikzpicture}
}

\newcommand{\blockedtwo}{
\begin{tikzpicture}
\node[shape=circle,draw,inner sep=2pt] (1) at (0,0) {};
\node[shape=circle,draw,inner sep=2pt] (2) at (1,0) {};
\draw[blocked] (1) to (2); 
\draw[track] (1) to (2);

\node[below] at (-0.06,-0.1) {\scriptsize \textbf{1}};
\node[below] at (1.06,-0.1) {\scriptsize \textbf{2}};

\node[
  draw,
  align=center,
  regular polygon,
  regular polygon sides=4,
  inner sep=-1pt, fill=white] (car1) at (0.5,0.5) {\faAutomobile};
\node[
  draw,
  align=center,
  regular polygon,
  regular polygon sides=4,
  inner sep=-1pt, fill=white] (car2) at (0.5,-0.5) {\faAutomobile};

\draw[thick, bend left] (1) to (car1);
\draw[thick, bend left] (car1) to (2);
\draw[thick, bend right] (1) to (car2);
\draw[thick, bend right] (car2) to (2);
\node[
  draw,semithick,
  align=center,
  regular polygon,
  regular polygon sides=3,
  inner sep=-2.3pt, fill=orange-red] at (0.5,-0.02) 
{\raisebox{1.1mm}{\textbf{$\lightning$}}};

\end{tikzpicture}
}

\usepackage{fontawesome}

\newenvironment{bem}{\begin{remark}}{\end{remark}}
\newenvironment{defi}{\begin{definition}}{\end{definition}}
\newenvironment{propos}{\begin{proposition}}{\end{proposition}}

\newenvironment{asm}{\begin{assumption}}{\end{assumption}}

\renewcommand{\V}{V}
\renewcommand{\E}{E}

\usepackage{microtype}
\usepackage{etoolbox}
\apptocmd{\sloppy}{\hbadness 10000\relax}{}{}

\usepackage{hyperref} %%hyperref zum schluss
\begin{document}
\renewcommand*{\thefootnote}{\fnsymbol{footnote}}
\title{Modeling Adverse Conditions in the Framework\\  of Graph Transformation Systems\footnote{This work is supported by the German Research Foundation through the Research Training Group DFG GRK 1765 SCARE: \emph{System Correctness under Adverse Conditions} (\href{https://www.uol.de/scare}{\texttt{https://www.uol.de/scare}}).}}
\author{Okan \"Ozkan
\institute{Department of Computing Science, \\ 
Carl von Ossietzky University of Oldenburg \\
Oldenburg, Germany}
\email{o.oezkan$@${informatik.uni-oldenburg.de}}}
\def\titlerunning{Modeling Adverse Conditions in the Framework of Graph Transformation Systems}
\def\authorrunning{Okan \"Ozkan}
\maketitle

\begin{abstract} \ignore{\scriptsize} The concept of adverse conditions addresses systems interacting with an adversary environment and finds use also \textcolor{new}{in the development of new technologies}. We present an approach for modeling adverse conditions by graph transformation systems. In contrast to other approaches for graph-transformational interacting systems, the presented main constructs are graph transformation systems. We introduce joint graph transformation systems which involve a system, an interfering environment, and an automaton modeling their interaction. \textcolor{new}{For joint graph transformation systems, we introduce notions of (partial) correctness under adverse conditions, which contain the correctness of the system and a recovery condition}. As main result, we show that two instances of correctness, namely $k$-step correctness (recovery in at most $k$ steps after an environment intervention) and last-minute correctness (recovery until next environment intervention) are expressible in LTL (linear temporal logic)\longv{, and that a weaker notion of $k$-step correctness is expressible in CTL (computation tree logic)}.
\end{abstract}

\ignore{%oldtitleOLD
\renewcommand*{\thefootnote}{\fnsymbol{footnote}}
\begin{center} \textbf{\Large{Modeling Adverse Conditions in the Framework}} \vspace*{3mm} \\ \textbf{\Large{of Graph Transformation Systems\footnote{This work is supported by the German Research Foundation through the Research Training Group DFG GRK 1765 SCARE: \emph{System Correctness under Adverse Conditions} (\texttt{www.uol.de/scare}).}}} \shortv{\vspace*{2mm}}\longv{\vspace*{3mm} \\ \large{\it Long Version}  \vspace*{6mm}} \\ \textbf{{Okan \"Ozkan}} \vspace*{2mm} \\ Carl von Ossietzky University of Oldenburg  \\ \texttt{okan.oezkan$@${uni-oldenburg.de} } \vspace*{1mm}   \\  \begin{abstract} \ignore{\scriptsize} \textbf{Abstract.} The concept of adverse conditions finds use in many areas. However, investigations on how to model adverse conditions by graph transformation systems are sparse. We introduce joint graph transformation systems which involve a system, an interfering environment, and a finite automaton regulating their interaction. Notions of (partial) correctness for joint graph transformation systems contain the correctness of the system and a recovery condition. As main result, we show that two instances of correctness, namely $k$-step correctness (recovery $k$ steps after an environment intervention) and last-minute correctness (recovery until next environment intervention) are expressible in LTL (linear temporal logic)\longv{, and that a weaker notion of $k$-step correctness is expressible in CTL (computation tree logic)}.
\end{abstract} \end{center}
}

\section{Introduction}
\label{intro}
\renewcommand*{\thefootnote}{\arabic{footnote}}
\setcounter{footnote}{0}
System correctness under \textit{adverse conditions} is a topic of recent research. This concept addresses systems which interact with an environment and finds use \ignore{e.g. in distributed systems \cite{Theel15} and }also in the development of new technologies such as autonomous driving \cite{Schwammberger18} and neural networks \cite{Worzyk19}.\footnote{For further topics, see, e.g.,\ the publication list of the Research Training Group DFG GRK 1765 SCARE (\href{https://www.uol.de/en/academic-research/collaborative-research-projects/scare/publications}{\texttt{https://www.uol.de/en/academic-research/collaborative-research-projects/scare/publications}}).} \emph{Correctness} in this sense means that the interaction between system and environment satisfies desired behavioral properties. One assumes that the environment exhibits an only partially predictable behaviour and may have an adverse effect. \par
\textit{Graph transformation systems (GTS)}, as considered, e.g.,\ in \cite{Ehrig06}, are a visual but yet precise formalism for modeling a system. In this perception, system states are captured by graphs and state changes by graph transformations.  \par
Topics in graph transformation are growing in popularity. While several concepts in the scope of graph-transformational interacting systems such as graph transformation systems with dependencies \cite{Corradini09}, distributed graph transformation \cite{Taentzer99}, autonomous units \cite{Kreowski06}, graph-transformational swarms \cite{Kreowski13}, and graph-transformational multi-agent systems \cite{Wang06} have been introduced, there is fewer research which explicitly considers the special case of adverse conditions, e.g.,\ \cite{Flick16, Peuser18}. Adverse conditions are modeled by an \textit{interfering environment}. Our approach focuses on the particular case in which the interaction between system and environment is regulated. A natural way for modeling the interaction is to utilize automata. The idea is to construct a \textit{joint system} which involves the system, the environment, and a regulation automaton. 
\begin{center} \tikzset{
    state/.style={
           rectangle,
           rounded corners,
           draw=black,
           minimum height=2em,
           inner sep=2pt,
           text centered,
           },
}

\hspace*{1.2cm}\begin{tikzpicture}[->,>=stealth']
\node[state,fill=whitegray] (box1) at (0,0.4) { \begin{tabular}{c} \hspace*{1mm} \\ construct \\ \hspace*{1mm} \end{tabular}};

\draw (-5,1) -- node[above]{system}(-1,1);  
\draw (-5,0.4) -- node[above]{environment}(-1,0.4);
\draw (-5,-0.2) -- node[above]{regulation automaton}(-1,-0.2); 

\draw (1,0.4) -- node[above]{joint system}(5,0.4); 

\end{tikzpicture} \vspace*{2mm}\newline
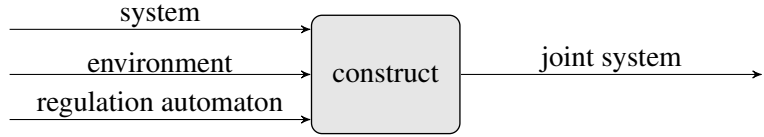
\hspace*{-6mm}Fig. 1: construction of the joint system \end{center}  \vspace*{2mm}
In contrast to other approaches, the constructed joint system is in the same class as the system and the environment, i.e.,\ a graph transformation system. This can be advantageous for transfering well-known concepts for graph transformation systems such as model checking \cite{GROOVE} and well-structuredness \cite{Koenig14} into this framework. \par Regarding the latter, the correctness notion for GTSs is also applicable to \textit{joint graph transformation systems}. The definition of (partial) correctness of a GTS w.r.t.\ a \emph{precondition} and a \emph{postcondition} means that for every graph satisfying the precondition, every graph derived via the GTS satisfies the postcondition. However, this turns out to be too restrictive \textcolor{new}{as a notion of correctness under adverse conditions}. The validity of the postcondition may be violated after an interference of the environment; nonetheless, the postcondition must be recovered. Correctness notions for joint GTSs contain the correctness of the system and a \textit{recovery condition}. We consider two instances of correctness notions, i.e.,\ recovery conditions: One possibility is to limit the maximal number of steps after which the postcondition must be recovered ($k$-\textit{step correctness}). By contrast, the notion of \textit{last-minute correctness} demands not a recovery after $k$ steps but a recovery until the environment interferes again. \ignore{These notions may also be seen as kinds of \emph{robustness} or \emph{resilience}.}\par
As main result, we show that correctness of joint GTSs in both cases can be formalized in terms of a temporal logic, namely LTL (\emph{linear temporal logic}). \longv{In addition to that, we show that a weaker notion of $k$-step correctness is expressible in CTL (\emph{computation tree logic}). }These reductions are favorable since, e.g.,\ the tool GROOVE \cite{GROOVE} provides a way for LTL\longv{/CTL} model checking for graph-based systems. \par
The concepts are illustrated by a simple traffic network system under stress. \begin{example}[traffic network system] \label{ex}We model a \emph{traffic network system (TNS)} where the adverse conditions are blocked connections due to jams, accidents, or damages on the tracks. Nodes correspond to traffic junctions while an edge of the form $\trackO$ represents a connection between them. An edge labeled with the symbol \faAutomobile\ describes a vehicle in the traffic line between the traffic junctions. The rule \texttt{Ascend}/\texttt{Descend} formalizes the ascent/descent of a vehicle onto/off the traffic network. We realize a traffic flow by the rule \texttt{Move}, which enables vehicles to change their position to an adjacent track. Blocked tracks ($\blockedO$) occur only if a connection is highly frequented, i.e.,\ if $v$ vehicles are on the same track; this is formalized by the environment rule \texttt{Block}. For the sake of simplicity, we choose $v=2$. Blocked tracks can be repaired by the rule \texttt{Repair}. Formally, we consider the following system $\Sys$ and environment $\Env$: \vspace*{3mm}\\  \ignore{ \vspace*{3mm} \begin{center} $\Sys \left\{ \begin{array}{ll} \texttt{Ascend} &: \left\tuple{ \track \Rightarrow \onecar \right} \\ \texttt{Descend} &: \left\tuple{ \onecar \Rightarrow \track \right} \\ \texttt{Move} &: \left\tuple{ \moveL \hspace*{2mm}\Rightarrow \hspace*{2mm}\moveR \right}  \\ \texttt{Repair} &: \left\tuple{ \blocked \Rightarrow \track \right} \end{array} \right.$ \\ \vspace*{2mm}\end{center}
\hspace*{35.5mm}$\Env \left\{  \begin{array}{ll}\texttt{Block}&\hspace*{4mm}: \left\tuple{ \twocar \Rightarrow \blockedtwo \right} \end{array} \right.$  \vspace*{3mm} }%thisold
\scalebox{0.93}{\begin{tabular}{c}\begin{minipage}[h]{0.6\textwidth} $\Sys \left\{ \begin{array}{ll} \texttt{Ascend} &: \left\tuple{ \track \Rightarrow \onecar \right} \\ \texttt{Descend} &: \left\tuple{ \onecar \Rightarrow \track \right} \\ \texttt{Move} &: \left\tuple{ \moveL \hspace*{2mm}\Rightarrow \hspace*{2mm}\moveR \right}  \\ \texttt{Repair} &: \left\tuple{ \blocked \Rightarrow \track \right} \end{array} \right.$ \end{minipage}\end{tabular}}  \hspace*{-1cm}\scalebox{0.93}{ \begin{tabular}{c}\vspace*{3.1cm}\begin{minipage}[h]{0.45\textwidth} $\Env \left\{  \begin{array}{ll}\texttt{Block}&\hspace*{5mm}: \left\tuple{ \twocar \Rightarrow \blockedtwo \right} \end{array} \right.$  \end{minipage} \end{tabular} } \begin{center}
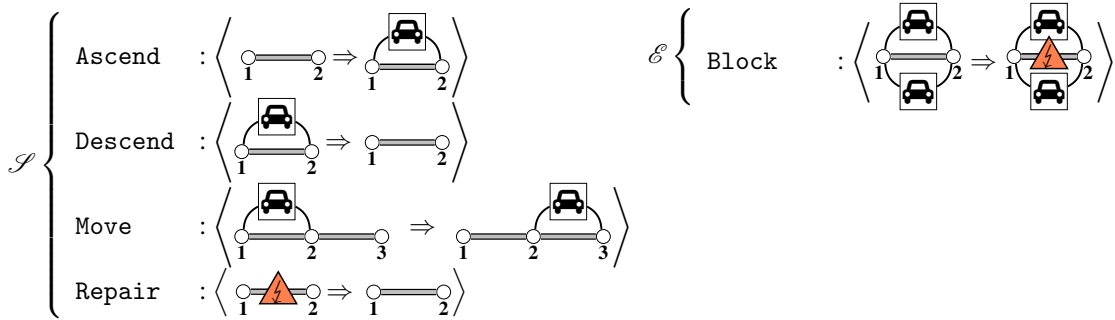
Fig. 2: system and environment rules of the TNS \end{center}
The regulation automaton $A$ modeling the interaction between system and environment is given by the following figure:  \vspace*{3mm}
\begin{center}\scalebox{0.9}{\begin{tikzpicture}[-> ,>= stealth' , node distance=3cm]

 \node[state,fill=whitegray] (R){$q_0$};
 \node[state, fill=whitegray] (S) [right of=R]{$q_1$};

\node (start) [left = 1cm of R] {};

\path (start) edge node{} (R) 
(R) edge[bend left,above] node{\texttt{Block}} (S)
(R) edge[loop above, above] node{\small{\{\texttt{Move}, \texttt{Ascend}, \texttt{Descend}\}}} (S)

 (S) edge[bend left, below] node {\texttt{Repair}} (R) ;
 \end{tikzpicture} }\end{center}  \begin{center} Fig. 3: representation of the regulation automaton $A$ \end{center} \vspace*{3mm} 
The joint system is a graph transformation system and constructed as the union of the ``enriched'' rule sets of system and environment. By this enrichment, the regulation automaton is ``synchronized'' with a rule set.   \end{example}
This paper is organized as follows: In Section \ref{prelim}, we recall the basic notions of graph transformation, graph conditions, \textcolor{new}{LTL, and CTL}\shortv{ (LTL)}, before defining joint graph transformation systems in Section \ref{joint}. In Section \ref{correct}, we investigate \longv{three}\shortv{two} correctness notions for joint GTSs. Our main result showing \shortv{that both correctness notions can be formalized in LTL is proven in Section \ref{reduct}.\footnote{A long version of this paper containing an instance of correctness, which is expressible in CTL, is available under: http://formale-sprachen.informatik.uni-oldenburg.de/{$\sim$}skript/fs-pub/ModelingAdverseConditionsLongVersion{\_}OOe.pdf} In Section \ref{related}, we present related concepts. We close with a conclusion and an outlook in Section~\ref{conc}.}\longv{that the presented correctness notions can be formalized in LTL/CTL is proven in Section \ref{reduct}. In Section \ref{related}, we present related concepts. We close with a conclusion and an outlook \textcolor{new}{in Section \ref{conc}}.} \apx{ In the Appendix, one can find complementary proofs and examples (\ref{subsequent}) and an instance of correctness for joint GTSs, which can be expressed in CTL (\ref{CTL}).}

\section{Preliminaries} \label{prelim}
We recall basic notions of graphs, graph conditions, rules, and transformations and define temporal graph constraints corresponding to temporal formulas.
\subsection{Graph Transformation Systems}

\newcommand\circlearound[1]{%
  \tikz[baseline]\node[draw,shape=circle,anchor=base] {#1} ;}

In the following, we recall the definitions of graphs, graph conditions, rules, and graph transformation systems \cite{Ehrig06}.\ignore{Ehrig-Ehrig-Prange-Taentzer06b,Habel-Pennemann09aHabel-Plump01a} \textcolor{new}{We assume that the reader is familiar with pushouts, see, e.g.,\ \cite{Ehrig06}.}

A directed, labeled graph consists of a set of nodes and a set of edges where each edge is equipped with a source and a target node and where each node and edge is equipped with a label.
\begin{definition}[graphs \& morphisms]  A \emph{(directed, labeled) graph} (over a label alphabet $\Lambda$) is a tuple $G=\tuple{\V_G,\E_G,\sou_G,\tar_G,\labb_{\V,G},\labb_{\E,G}}$ where $\V_G$ and $\E_G$ are finite sets of \emph{nodes} (or \emph{vertices}) and \emph{edges}, 
$\sou_G,\tar_G\colon$ $\E_G\to \V_G$ are functions assigning \emph{source} and \emph{target}\/ to each edge, and $\labb_{\V,G}\colon\V_G\to\Lambda$, $\labb_{\E,G}\colon\E_G\to\Lambda$ are labeling functions. Given graphs $G$ and $H$, a \emph{(graph) morphism} $g\colon G \to H$ consists of functions $g_\V\colon\V_G\to\V_H$ and $g_\E\colon\E_G\to\E_H$ that preserve sources, targets, and labels, i.e., $g_\V\circ\sou_G=\sou_H\circ g_\E$, $g_\V\circ\tar_G=\tar_H\circ g_\E$, $\labb_{\V,G}=\labb_{\V,H}\circ g_\V$, $\labb_{\E,G}=\labb_{\E,H}\circ g_\E$. The morphism $g$ is \emph{injective (surjective)}\/ if $g_{\V}$ and $g_{\E}$ are injective (surjective), and an \emph{isomorphism}\/ if it is injective and surjective. In the latter case, $G$ and $H$ are \emph{isomorphic}, which is denoted by $G\cong H$. If $g$ is injective, we write also $g: G \injto H$. The composition of morphisms is defined componentwise.
\end{definition}
Graph conditions are nested constructs which can be represented as trees of morphisms equipped with quantifiers and Boolean connectives. Graph conditions and first-order graph formulas are expressively equivalent \cite{Rensink04,Habel09}. \newpage
\begin{definition}[graph conditions] The class of \emph{(graph) conditions} over a graph $P$ is defined inductively: (i) $\ctrue$ is a graph condition over $P$, (ii) $\exists(a,c)$ is a graph condition over $P$ where $a\colon P \injto C$ is an injective morphism and $c$ is a condition over $C$, (iii) for conditions $c$, $c'$ over $P$, $\neg c$ and $ c\land c'$ are conditions over~$P$. \\ Conditions over the empty graph~$\emptyset$ are called \emph{constraints}. In the context of rules, conditions are called \emph{application conditions}. \end{definition}
Graph conditions may be written in a more compact form: $\PE a$ abbreviates $\PE(a,\ctrue)$, $\cfalse$ abbreviates $\neg \ctrue$\ignore{ and $\PA(a,c)$ abbreviates $\neg\PE(a, \neg c)$}. The expressions $c \vee c'$ and $c\impl c'$ are defined as usual. For an injective morphism $a\colon P\DSinjto C$ in a condition, we just depict the codomain $C$ if the domain $P$ can be unambiguously inferred.
\begin{example}[no blocked track] The constraint $\text{NoBlocked}:=\neg \exists \left(  \blockedO \right) $ expresses that, intuitively speaking, there is no blocked track in the traffic network (see Example \ref{ex}).\end{example}
\ignore{%old
\begin{definition}[semantics] The semantics of graph conditions are defined inductively: (i) Any injective morphism $p\colon P\injto G$ \emph{satisfies} $\ctrue$. \\
\begin{tabular}{c}\hspace*{-2mm}\begin{minipage}[h]{10cm}
 (ii) An injective morphism $p$ \emph{satisfies} $\PE(a,c)$ with $a\colon P\injto C$ if there exists an injective morphism $q\colon C\injto G$ such that $q\circ a=p$ and $q$~satisfies~$c$. (iii) An injective morphism $p$ \emph{satisfies} $\neg c$ if $p$ does not satisfy $c$, and $p$ \emph{satisfies} $c \land c'$ if $p$ satisfies both $c$ and $c'$. 
\end{minipage}\end{tabular}
\hspace{0.5cm}
\begin{tabular}{c}\begin{minipage}[h]{4cm}
\tikz[node distance=2em,shape=rectangle,outer sep=1pt,inner sep=2pt]{
\node(P){$P$};
\node(G)[strictly below right of=P]{$G$};
\node(C)[strictly above right of=G]{$C,$};
\draw[monomorphism] (P) -- node[overlay,above](a){$a$} (C);
\draw[monomorphism] (P) -- node[overlay,below left]{$p$} (G);
\draw[altmonomorphism] (C) -- node[overlay,below right](q){$q$} (G);
\draw[draw=white] (a) -- node[overlay](tr1){=} (G);
\node(c)[outer sep=0pt,inner sep=0pt,node distance=0em,strictly right of=C]{\tikz[draw=black,fill=lightgray]{
\filldraw (0,0) -- (0.6,0.12) -- node[right,outer sep=1ex]{\footnotesize{$c$}} (0.6,0) -- (0.6,-0.12) -- (0,0);}};
\draw[draw=white] (q) -- node[overlay,sloped](tr1){$\models$} (c);
\node(Y)[node distance=0.2em,strictly right of=c]{$)$};
\node(X)[node distance=0.0em,strictly left of=P]{$\PE($};}
\end{minipage}\end{tabular}
We write  $p\models c$ if $p$ satisfies the condition $c$ (over $P$). A graph $G$ \emph{satisfies} a constraint $c$, $G\models c$, if the morphism $p\colon\emptyset\injto G$ satisfies~$c$. 
\end{definition} }

\begin{definition}[semantics] The semantics of graph conditions are defined inductively: (i) Any injective morphism $p\colon P\injto G$ \emph{satisfies} $\ctrue$. \\
 (ii) An injective morphism $p$ \emph{satisfies} $\PE(a,c)$ with $a\colon P\injto C$ if there exists an injective morphism $q\colon C\injto G$ such that $q\circ a=p$ and $q$~satisfies~$c$. 
\begin{center}\tikz[node distance=2em,shape=rectangle,outer sep=1pt,inner sep=2pt]{
\node(P){$P$};
\node(G)[strictly below right of=P]{$G$};
\node(C)[strictly above right of=G]{$C,$};
\draw[monomorphism] (P) -- node[overlay,above](a){$a$} (C);
\draw[monomorphism] (P) -- node[overlay,below left]{$p$} (G);
\draw[altmonomorphism] (C) -- node[overlay,below right](q){$q$} (G);
\draw[draw=white] (a) -- node[overlay](tr1){=} (G);
\node(c)[outer sep=0pt,inner sep=0pt,node distance=0em,strictly right of=C]{\tikz[draw=black,fill=lightgray]{
\filldraw (0,0) -- (0.6,0.12) -- node[right,outer sep=1ex]{\footnotesize{$c$}} (0.6,0) -- (0.6,-0.12) -- (0,0);}};
\draw[draw=white] (q) -- node[overlay,sloped](tr1){$\models$} (c);
\node(Y)[node distance=0.2em,strictly right of=c]{$)$};
\node(X)[node distance=0.0em,strictly left of=P]{$\PE($};}\end{center}
 (iii) An injective morphism $p$ \emph{satisfies} $\neg c$ if $p$ does not satisfy $c$, and $p$ \emph{satisfies} $c \land c'$ if $p$ satisfies both $c$ and $c'$. \\
We write  $p\models c$ if $p$ satisfies the condition $c$ (over $P$). A graph $G$ \emph{satisfies} a constraint $c$, $G\models c$, if the morphism $p\colon\emptyset\injto G$ satisfies~$c$. 
\end{definition}

\begin{bem} The validity of graph constraints is closed
under isomorphisms, i.e.,\ for every graph constraint $c$ and every
isomorphism $G\cong G'$, $G\models c$ iff $G'\models
c$. This can be shown by an induction over the structure of constraints. \end{bem}
\ignore{$\PE(x,\ctrue)\equiv\PE x$\\
$\PA(x,\ctrue)\equiv\ctrue$\\
$\PE(x,\cfalse)\equiv\cfalse$\\
$\PA(x,\false)\equiv\NE x$\\
$\PA(x,\PE(y,\cfalse)\equiv\PA(x,\cfalse)\equiv\NE x$\\
$\PE x,\PA(x,\false)\equiv\PE(x,\NE y)$\\}
Rules are specified by a pair of injective morphisms. For restricting the applicability of rules,  the rules are  equipped with a left application condition. Such a rule is applicable with respect to an injective ``match'' morphism from the left-hand side of the rule to a graph iff  the underlying plain rule is applicable and the match morphism satisfies the left application condition.

%%%Version
\ignore{
\begin{definition}[graph transformation rules and transformations] A~\emph{(graph transformation) rule}\/ $\prule = \tuple{p,\ac}$ consists  of a \emph{plain rule} $p=\brule{L}{K}{R}$ with injective morphisms $K\injto L$ and $K\injto R$ and an application condition $\ac$ over $L$. A rule $\tuple{p,\ctrue}$ is abbreviated by~$p$.
A \emph{direct transformation}\/ from a graph $G$ to a graph $H$ applying rule $\prule$ at an injective morphism $g$ consists of two pushouts\footnote{For definition \& existence of pushouts in the category of graphs, see e.g. \cite{Ehrig06}.} (1) and (2) as above where $g\models\ac$. 

\begin{tikzpicture}[node distance=2.5em,shape=rectangle,outer sep=1pt,inner sep=2pt,label distance=-1.25em]
\node(L){$L$};
\node(K)[strictly right of=L]{$K$};
\node(R)[strictly right of=K]{$R$};
\node(D)[strictly below of=K]{$D$};
\node(G)[strictly left of=D]{$G$};
\node(H)[strictly right of=D]{$H$};
%horizontal morphisms
\draw[altmonomorphism] (K) -- node[overlay,above]{\small $$} (L);
\draw[monomorphism] (K) -- node[overlay,above]{\small $$} (R);
\draw[altmonomorphism] (D) -- (G);
\draw[monomorphism] (D) -- (H);
%vertical morphisms 
\draw[monomorphism] (L) -- node[overlay,left](g){\small $g$} (G);
\draw[monomorphism] (K) -- node[overlay,left]{\small $d$} (D);
\draw[monomorphism] (R) -- node[overlay,right]{\small $h$}(H);
\draw[draw=none] (L) -- node[overlay]{\small (1)} (D);
\draw[draw=none] (R) -- node[overlay]{\small (2)} (D);
\node(acL)[outer sep=0pt,inner sep=0pt,node distance=0em,strictly left of=L]{
  \tikz[baseline,draw=black,fill=lightgray]{\filldraw (0,0) -- node[left,pos=0.9,overlay,outer sep=1em](acL){\small $\ac$} 
  (-0.6,0.12) -- (-0.6,-0.12) -- (0,0);}};
\draw[draw=none] (g) -- node[overlay,sloped](tr1){$\mathrel{=}\joinrel\mathrel{|}$} (acL);
 \end{tikzpicture}
We write $G\dder_{\prule,g} H$ or $G\dder_{\prule} H$ if there exists such a direct transformation and $G\not\dder_\prule$ if there is no graph $H$ such that $G\dder_\prule H$. Given graphs $G$, $H$\/  and a set $\R$\/ of rules, $G$ \emph{derives} $H$\/ by $\R$\/ if $G \cong H$ or there is a sequence of direct transformations $G=G_0\dder_{\prule_1}\dots \dder_{\prule_n}G_n=H$ with $\prule_1,\dots,\prule_n\in \R$. In this case we write $G\der_{\R} H$\/ or just $G \der H$.
\end{definition}
}

\ignore{%oldversion
\begin{definition}[rules \& transformations]
A~\emph{(graph transformation) rule} ~$r =~\tuple{p,\ac}$ (over $\Lambda$) con-
\begin{tabular}{c}\hspace*{-2mm}\begin{minipage}[h]{10cm}sists  of a \emph{plain rule} ~$p=\tuple{L \hookleftarrow K \injto R}$ with injective morphisms  $K\injto L$ and $K\injto R$ and an application condition $\ac$ over $L$. A rule $\tuple{p,\ctrue}$ is abbreviated by $p$ (where $L$, $K$, and $R$ are graphs over $\Lambda$). A \emph{(direct) transformation}\/ from a graph $G$ to a graph $H$ applying rule $\prule$ at an injective morphism $g$ consists of two \emph{pushouts}\ignore{\footnote{For definition \& existence of pushouts in the category of graphs, see e.g.,\ \cite{Ehrig06}.}} (1) and (2) as shown in the figure where $g\models\ac$ (for the definition \&
\end{minipage}\end{tabular}
\hspace{0.5cm}
\begin{tabular}{c}\hspace*{2mm}\begin{minipage}[h]{3.4cm}
\begin{tikzpicture}[node distance=2.5em,shape=rectangle,outer sep=1pt,inner sep=2pt,label distance=-1.25em]
\node(L){$L$};
\node(K)[strictly right of=L]{$K$};
\node(R)[strictly right of=K]{$R$};
\node(D)[strictly below of=K]{$D$};
\node(G)[strictly left of=D]{$G$};
\node(H)[strictly right of=D]{$H$};
%horizontal morphisms
\draw[altmonomorphism] (K) -- node[overlay,above]{\small $$} (L);
\draw[monomorphism] (K) -- node[overlay,above]{\small $$} (R);
\draw[altmonomorphism] (D) -- (G);
\draw[monomorphism] (D) -- (H);
%vertical morphisms 
\draw[monomorphism] (L) -- node[overlay,left](g){\small $g$} (G);
\draw[monomorphism] (K) -- node[overlay,left]{\small $d$} (D);
\draw[monomorphism] (R) -- node[overlay,right]{\small $h$}(H);
\draw[draw=none] (L) -- node[overlay]{\small (1)} (D);
\draw[draw=none] (R) -- node[overlay]{\small (2)} (D);
\node(acL)[outer sep=0pt,inner sep=0pt,node distance=0em,strictly left of=L]{
  \tikz[baseline,draw=black,fill=lightgray]{\filldraw (0,0) -- node[left,pos=0.9,overlay,outer sep=1em](acL){\small $\ac$} 
  (-0.6,0.12) -- (-0.6,-0.12) -- (0,0);}};
\draw[draw=none] (g) -- node[overlay,sloped](tr1){$\mathrel{=}\joinrel\mathrel{|}$} (acL);
 \end{tikzpicture}
\end{minipage}\end{tabular}

existence of pushouts in the category of graphs, see, e.g.,\ \cite{Ehrig06}). By $G\dder H$ we denote (the existence of) a direct transformation from $G$ to $H$. A rule $r$ is \emph{applicable} to a graph $G$ if there is a transformation from $G$ to a graph $H$ via $r$. \\ A \emph{transformation sequence} is a sequence $\tuple{G=G_0 \dder \ldots (\dder G_n =H)}$ of direct transformations. We call $G_0$ the \emph{starting graph} and $n$ the \emph{length}. For $n \ge 0$, we write $G \dder^* H$ or $G \dder^n H$. For $n\ge 1$, we denote this also by $G \dder^+ H$. (By $n=0$, we mean $G\cong H$.) If there is an $n \le k$ s.t. $G \dder^n H$, we write also $G \dder^{\le k} H$. \\Direct transformations and transformation sequences may be indicated with the applied rule or the rule set, respectively, to which the applied rules belong.
\end{definition}}
\begin{definition}[rules \& transformations]
A~\emph{(graph transformation) rule} ~$r =~\tuple{p,\ac}$ (over $\Lambda$) consists  of a \emph{plain rule} ~$p=\tuple{L \hookleftarrow K \injto R}$ with injective morphisms  $K\injto L$ and $K\injto R$ and an application condition $\ac$ over $L$. A rule $\tuple{p,\ctrue}$ is abbreviated by $p$ (where $L$, $K$, and $R$ are graphs over $\Lambda$). A \emph{(direct) transformation}\/ from a graph $G$ to a graph $H$ applying rule $\prule$ at an injective morphism $g$ consists of two \emph{pushouts}\ignore{\footnote{For definition \& existence of pushouts in the category of graphs, see e.g.,\ \cite{Ehrig06}.}} (1) and (2) as shown in the figure where $g\models\ac$ (for the definition \& existence of pushouts in the category of graphs, see, e.g.,\ \cite{Ehrig06}). By $G\dder H$ we denote (the existence of) a direct transformation from $G$ to $H$. A rule $r$ is \emph{applicable} to a graph $G$ if there is a transformation from $G$ to a graph $H$ via $r$. \vspace*{3mm}\begin{center}\hspace*{-10mm}\begin{tikzpicture}[node distance=2.5em,shape=rectangle,outer sep=1pt,inner sep=2pt,label distance=-1.25em]
\node(L){$L$};
\node(K)[strictly right of=L]{$K$};
\node(R)[strictly right of=K]{$R$};
\node(D)[strictly below of=K]{$D$};
\node(G)[strictly left of=D]{$G$};
\node(H)[strictly right of=D]{$H$};
%horizontal morphisms
\draw[altmonomorphism] (K) -- node[overlay,above]{\small $$} (L);
\draw[monomorphism] (K) -- node[overlay,above]{\small $$} (R);
\draw[altmonomorphism] (D) -- (G);
\draw[monomorphism] (D) -- (H);
%vertical morphisms 
\draw[monomorphism] (L) -- node[overlay,left](g){\small $g$} (G);
\draw[monomorphism] (K) -- node[overlay,left]{\small $d$} (D);
\draw[monomorphism] (R) -- node[overlay,right]{\small $h$}(H);
\draw[draw=none] (L) -- node[overlay]{\small (1)} (D);
\draw[draw=none] (R) -- node[overlay]{\small (2)} (D);
\node(acL)[outer sep=0pt,inner sep=0pt,node distance=0em,strictly left of=L]{
  \tikz[baseline,draw=black,fill=lightgray]{\filldraw (0,0) -- node[left,pos=0.9,overlay,outer sep=1em](acL){\small $\ac$} 
  (-0.6,0.12) -- (-0.6,-0.12) -- (0,0);}};
\draw[draw=none] (g) -- node[overlay,sloped](tr1){$\mathrel{=}\joinrel\mathrel{|}$} (acL);
 \end{tikzpicture}\end{center}
\vspace*{3mm}A \emph{transformation sequence} (of \emph{length} $n$) is a sequence $\tuple{G=G_0 \dder \ldots (\dder G_n =H)}$ of direct transformations. For $n \ge 0$, we write $G \dder^* H$ or $G \dder^n H$. For $n\ge 1$, we denote this also by $G \dder^+ H$. (By $n=0$, we mean $G\cong H$.) If there is an $n \le k$ s.t. $G \dder^n H$, we write also $G \dder^{\le k} H$. \\Direct transformations and transformation sequences may be indicated with the applied rule or the rule set, respectively, to which the applied rules belong.
\end{definition}

\begin{notation}A rule $\brule{L}{K}{R}$ sometimes is denoted by $\tuple{ L\dder R }$ where indices in $L$ and $R$ refer to the corresponding nodes. \end{notation}
We consider the simplest, unregulated type of system, i.e.,\ a finite set of graph transformation rules.
\begin{defi}[graph transformation system] A \emph{graph transformation system (GTS)} (over $\Lambda$) is a finite set of graph transformation rules (over $\Lambda$). \end{defi}
Completing finite sequences to infinite sequences is one way to handle the semantics of temporal graph constraints since temporal formulas are usually defined on infinite sequences. The following proposition provides the construction of the \emph{completion} of a GTS $\R$ \textcolor{new}{where finite sequences are extended to infinite sequences}. \shortv{It follows from\napx{ Theorem 1 and Fact 2 of} \cite{Pennemann06}. \apx{For details, see Appendix \ref{app}.}}
\begin{proposition}[applicability] \label{applic}For every graph transformation system $\R$ and every plain rule $p=\tuple{L \dder R}$, there is an application condition $\app(\R,p)$ s.t.\ for every transformation $G \dder_p H'$ with $g: L \injto G$, $ g \models \app(\R,p)\hspace*{2mm}\text{iff}\hspace*{2mm} \exists H: G \dder_{\R} H$. Consequently, the rule $\tuple{p, \neg\app(\R,p)}$ is applicable only if no rule from $\R$ is applicable.   \end{proposition}

\ignore{%oldvers
 \longv{\begin{proof}For the plain rule $p=\brule{L}{K}{R} $, let $p^{-1}=$$ ~\brule{R}{K}{L}$ be the \emph{inverse rule}. Let \begin{displaymath} \app(\R, p):= \text{A}(p^{-1}, \vee_{r \in \R} \neg\text{wlp}(r, \texttt{false})).  \end{displaymath} 
\begin{tabular}{c}\begin{minipage}[h]{7.2cm}
Consider the transformation  $H' \dder_{p^{-1}} G$.
\end{minipage}\end{tabular}
\hspace{0.5cm}
\begin{tabular}{c}\begin{minipage}[h]{3.4cm}
\begin{tikzpicture}[node distance=2.5em,shape=rectangle,outer sep=1pt,inner sep=2pt,label distance=-1.25em]
\node(L){$L$};
\node(K)[strictly right of=L]{$K$};
\node(R)[strictly right of=K]{$R$};
\node(D)[strictly below of=K]{$D$};
\node(G)[strictly left of=D]{$G$};
\node(H)[strictly right of=D]{$H'$};
%horizontal morphisms
\draw[altmonomorphism] (K) -- node[overlay,above]{\small $$} (L);
\draw[monomorphism] (K) -- node[overlay,above]{\small $$} (R);
\draw[altmonomorphism] (D) -- (G);
\draw[monomorphism] (D) -- (H);
%vertical morphisms 
\draw[monomorphism] (L) -- node[overlay,left](g){\small $g$} (G);
\draw[monomorphism] (K) -- node[overlay,left]{\small $d$} (D);
\draw[monomorphism] (R) -- node[overlay,right]{\small $h'$}(H);
\draw[draw=none] (L) -- node[overlay]{} (D);
\draw[draw=none] (R) -- node[overlay]{} (D);
\node(acL)[outer sep=0pt,inner sep=0pt,node distance=0em,strictly left of=L]{
  \tikz[baseline,draw=black,fill=lightgray]{\filldraw (0,0) -- node[left,pos=0.9,overlay,outer sep=1em](acL){\scriptsize $\text{app}(\R,p)$} 
  (-0.6,0.12) -- (-0.6,-0.12) -- (0,0);}};
\draw[draw=none] (g) -- node[overlay,sloped](tr1){$\mathrel{=}\joinrel\mathrel{|}$} (acL);
 \end{tikzpicture}
\end{minipage}\end{tabular}
 \begin{align*}g \models  \app(\R, p)& & \textnormal{iff} &\\g \models \text{A}(p^{-1}, \vee_{r \in \R} \neg\text{wlp}(r, \texttt{false})) & & \textnormal{iff}& \hspace*{2mm}\text{(Theorem 1 of \cite{Pennemann06})}    \\ G \models \vee_{r \in \R} \neg\text{wlp}(r, \texttt{false})& & \textnormal{iff}& \hspace*{2mm}\text{(Fact 2 of \cite{Pennemann06})} \\ \exists r \in \R: \exists H: G \dder_{r} H &&& \end{align*}
If $\tuple{p, \neg\text{app}(\R,p)}$ is applicable to a graph $G$, there is a morphism $g: L \injto G$ s.t. $g \models \neg\text{app}(\R,p)$, i.e. there is no graph $H$ s.t. $G \dder_{\R}H$. \end{proof}} }

\begin{proof} For the plain rule $p=\brule{L}{K}{R} $, let $p^{-1}=$$ ~\brule{R}{K}{L}$ be the \emph{inverse rule} of $p$. Let  \begin{displaymath} \app(\R, p):= \text{A}(p^{-1}, \vee_{r \in \R} \neg\text{wlp}(r, \texttt{false})).  \end{displaymath} 
Consider the transformation  $H' \dder_{p^{-1}} G$: 
\begin{center}\begin{tikzpicture}[node distance=2.5em,shape=rectangle,outer sep=1pt,inner sep=2pt,label distance=-1.25em]
\node(L){$L$};
\node(K)[strictly right of=L]{$K$};
\node(R)[strictly right of=K]{$R$};
\node(D)[strictly below of=K]{$D$};
\node(G)[strictly left of=D]{$G$};
\node(H)[strictly right of=D]{$H'$};
%horizontal morphisms
\draw[altmonomorphism] (K) -- node[overlay,above]{\small $$} (L);
\draw[monomorphism] (K) -- node[overlay,above]{\small $$} (R);
\draw[altmonomorphism] (D) -- (G);
\draw[monomorphism] (D) -- (H);
%vertical morphisms 
\draw[monomorphism] (L) -- node[overlay,left](g){\small $g$} (G);
\draw[monomorphism] (K) -- node[overlay,left]{\small $d$} (D);
\draw[monomorphism] (R) -- node[overlay,right]{\small $h'$}(H);
\draw[draw=none] (L) -- node[overlay]{} (D);
\draw[draw=none] (R) -- node[overlay]{} (D);
\node(acL)[outer sep=0pt,inner sep=0pt,node distance=0em,strictly left of=L]{
  \tikz[baseline,draw=black,fill=lightgray]{\filldraw (0,0) -- node[left,pos=0.9,overlay,outer sep=1em](acL){\scriptsize $\text{app}(\R,p)$} 
  (-0.6,0.12) -- (-0.6,-0.12) -- (0,0);}};
\draw[draw=none] (g) -- node[overlay,sloped](tr1){$\mathrel{=}\joinrel\mathrel{|}$} (acL);
 \end{tikzpicture}\end{center}
 \begin{align*}g \models  \app(\R, p)& & \textnormal{iff} &\\g \models \text{A}(p^{-1}, \vee_{r \in \R} \neg\text{wlp}(r, \texttt{false})) & & \textnormal{iff}& \hspace*{2mm}\text{(Theorem 1 of \cite{Pennemann06})}    \\ G \models \vee_{r \in \R} \neg\text{wlp}(r, \texttt{false})& & \textnormal{iff}& \hspace*{2mm}\text{(Fact 2 of \cite{Pennemann06})} \\ \exists r \in \R: \exists H: G \dder_{r} H &&& \end{align*}
If $\tuple{p, \neg\text{app}(\R,p)}$ is applicable to a graph $G$, there is a morphism $g: L \injto G$ s.t.\ $g \models \neg\text{app}(\R,p)$, i.e.,\ there is no graph $H$ s.t. $G \dder_{\R}H$. \end{proof}

\begin{bem} The rule $\tuple{p, \neg\text{app}(\R,p)}$ is not applicable to a graph if $p \in \R$. We consider the special case $p=\texttt{Skip} \notin \R$. In this case, $ \tuple{\texttt{Skip}, \neg\text{app}(\R, \texttt{Skip})} $ is applicable iff no rule from $\R$ is applicable.
\end{bem}

\begin{defi}[completion] For a graph transformation system $\R$, the \emph{completion} $\overline{\R}$ of $\R$ is given by $\overline{\R} := \R \cup \{ \tuple{\texttt{Skip}, \neg\text{app}(\R, \texttt{Skip})} \}$. \end{defi}
\textcolor{new}{The completion of a GTS yields only infinite transformation sequences by  applying the rule $\texttt{Skip}:=\tuple{\varnothing \dder \varnothing}$ when no rule from $\R$ is applicable.}
\subsection{Temporal Graph Constraints}

Temporal formulas such as LTL \longv{and CTL} formulas are well-known in logics, see, e.g.,\ \cite{\longv{Clarke82,}Emerson90, Baier08}. We adapt the notions and consider so-called \emph{LTL}\longv{/\emph{CTL}} \emph{graph constraints} \cite{Peuser18}, i.e.,\ \shortv{LTL}\longv{temporal} formulas whose atoms equate to graph constraints. \longv{ \par We first define LTL graph constraints and their semantics.} The temporality is interpreted as the changes along a transformation sequence. Every direct transformation correlates to a time step. Besides the common propositional operators there are \emph{temporal operators}, e.g.,\ the operator $\opX$ (\emph{Ne\textbf{X}t}) describes the validity of a formula in the next step while $\opG$ (\emph{\textbf{G}lobally}) describes the validity of a formula in every following step. The operator $\opU$ (\emph{\textbf{U}ntil}) describes the validity of a first formula until a second formula is valid.

\begin{defi}[LTL graph constraints \& semantics] The class $\LTL$ of \textit{linear temporal logic \textcolor{new}{(graph)} constraints} is defined inductively: (i) Every graph constraint is in $\LTL$, (ii) for all $\phi , \psi \in \LTL$, $\phi \land \psi , $ \linebreak $\phi \lor \psi, \neg \phi, \phi \Rightarrow \psi \in \LTL$ (\emph{propositional operators}), (iii) for all $\phi , \psi \in \LTL$, $\opX \phi, \opG \phi, \phi \opU \psi,  \phi \opW \psi \in \LTL$ (\emph{temporal operators}).\\
The semantics of LTL constraints is defined for infinite transformation sequences: Let $S = \tuple{G_0 \dder  \ldots }$ be an infinite transformation sequence. The \textit{satisfaction} of LTL constraints in $S$, denoted by $\models$, is defined inductively: \\ (i) For a graph constraint $c$, $G_i \models c$ if $G_i$ satisfies $c$ as graph constraint. \\ (ii) The semantics for the propositional operators are as usual, e.g.,\ for $\phi \in \LTL$, $G_i \models_{\LTLt} \neg \phi$ if $G_i \not\models \phi$. \\ (iii) For all $\phi,\psi \in \LTL$, \\ \hspace*{5mm} (a) $G_i \models \opX \phi $ if $G_{i+1} \models_{\LTLt} \phi$, \ignore{\hspace*{25.3mm}}\\ \hspace*{4.95mm} (b) $G_i \models_{\LTLt} \opG \phi $ if $G_k \models_{\LTLt} \phi $ for all $k \ge i$, \\ \hspace*{5mm}  (c) $G_i \models \phi \opU \psi $ if there is $l \ge i$ s.t. $G_l \models_{\LTLt} \psi$ and $G_k \models_{\LTLt} \phi$ for all $i \le k < l$, \\ \hspace*{4.95mm} (d) $G_i \models \phi \opW \psi$ if $G_i \models \phi \opU \psi \lor \opG \phi$. \\
For an LTL constraint $\phi$, $S$ \emph{satisfies} $\phi$, in symbols $S \models \phi$, if $G_0 \models \phi$. \\ A GTS $\R$ \emph{satisfies} $\phi$, in symbols $\R \models \phi$, if $S \models \phi $ for all infinite sequences $S$ in the completion $\overline{\R}$. \ignore{, and  $F \models \phi $ (respectively $\ext (F) \models \st ( \phi )$) for all terminating sequences $F$ in $\R$.} \end{defi}

\longv{
\begin{bem} It is not necessary to distinguish between the propositional operators of LTL and the propositional operators in the graph constraints since they are ``equivalent'', e.g.,\ $G \models a \land b $ \textcolor{new}{iff $G \models a \land_{\text{LTL}} b $ for every graph $G$ and all graph constraints $a,b$ where $\land_{\text{LTL}}$ is the conjunction of LTL}. \end{bem}}
\longv{CTL graph constraints are, like LTL graph constraints, temporal formulas where the atoms equate to the graph constraints. By contrast, the temporality is here branching, i.e.,\ we consider changes along a transformation \emph{tree}. Besides the common propositional operators, there are \emph{path-quantified temporal operators} which are pairs of operators: the first one is either $\opA$ (\emph{for} \emph{\textbf{A}ll following paths}) or $\opE$ (\emph{there} \emph{\textbf{E}xists a following path}), the second one is a temporal operator. The operator $\opAG$ means the valdity of a formula in all following sequences, and $\opEF$ means that a graph can be reached where the formula is valid.

\begin{defi}[CTL graph constraints] The class $\CTL$ of \textit{computation tree logic \textcolor{new}{(graph)} constraints} is defined inductively by: (i) Every graph constraint is an element of $\CTL$. (ii) For every $\phi , \psi \in \LTL$: (a) $\phi \land \psi ,\phi \lor \psi, \neg \phi, \phi \Rightarrow \psi \in \CTL$ (\emph{propositional operators}), (b) $\opAX \phi, \opEX \phi, \opAG \phi,\opEG \phi, \phi \opAU \psi,\phi \opEU \psi,$ \linebreak $\phi \opAW \psi, \phi\opEW \psi \in \CTL$ (\emph{path-quantified temporal operators}).  \end{defi}
The semantics of CTL constraints is defined on a transformation tree which is constructed as the \emph{unfolding} of a GTS with a distinguished root graph. We give a formal definition of transformation trees.

\ignore{%oldv
\begin{defi}[transformation tree] Let $R$ be a GTS and $G$ be a graph. A \emph{transformation tree} $T(G)$ is a tree $\tuple{G, \dder_\R}$ with the \emph{root} $G$. The \emph{predecessor function} $\pred : \mathcal{L}(G) \to \mathcal{L}(G)$ assigns every graph $H$ in the \emph{language}\footnote{The language of a graph $G$ is the set of all graphs which are reachable from $G$ via $\dder_\R$, see \cite{Ehrig06}.} of $G$, a graph from which we can directly reach $H$. The \emph{successor function} $\pred : \mathcal{L}(G) \to \Pot (\mathcal{L}(G))$ assigns every graph $H \in \mathcal{L}(G)$ the set $\suc (H)= \{ M : H \dder_\R  M \}$ of all graphs which can be reached directly from $H$. A \emph{path} of $T(G)$ is a sequence in $\dder_\R$, which is also a subtree of $T(G)$.  \end{defi}}

\begin{defi}[transformation tree]  Let $\R$ be a GTS and $G$ be a graph. The \emph{transformation tree} $T(G)$ of $G$ w.r.t. $\R$ is a tuple $\tuple{ \FinSeq, \subset_\R}$ where $\FinSeq $ is the set of finite transformation sequences in $\R$ starting from $G$ and $\subset_\R$ is a relation on $\FinSeq$ s.t.\ $F \subset_\R F'$ iff $F$ is a subsequence of $F'$ and their lengths differ by $1$. A \emph{path} in $T(G)$ is a sequence of consecutively $\subset_\R$-related transformation sequences in $\FinSeq$. \ignore{The sequence $\tuple{G}$, respectively, $G$ is called the $root$ of $T(G)$.}\end{defi}
By identifying a transformation sequence $\tuple{G \dder_\R \ldots \dder_\R H} \in \FinSeq$ with the graph $H$ and $\subset_\R$ with $\dder_\R$, we obtain the following characterization of paths in the transformation tree: A path in $T(G)$ is a transformation sequence in $\R$ starting from a graph $H \in  \mathcal{L}(G)$ where $\mathcal{L}(G)$ is the \emph{language} of $G$, i.e.,\ the set of all graphs which are reachable from $G$ via $\dder_\R$, see, e.g.,\ \cite{Ehrig06}.

\begin{defi}[semantics of CTL] Let $T(G)$ be a transformation tree. The \textit{satisfaction} of CTL constraints in $T(G)$ is defined inductively: \\ (i) For a graph constraint $c$ and $H \in \mathcal{L}(G)$,  $H \models c$ if $H$ satisfies $c$ as graph constraint. \\(ii) The semantics for the propositional operators are as usual, e.g.,\ for $\phi \in \CTL$ and $H \in \mathcal{L}(G)$,  $H \models \neg \phi$ if $H \not\models \phi$. \\ (iii) For $\phi, \psi \in \CTL$, $H \in \mathcal{L}(G)$ and $\Op \in \{ \opX, \opG, \opU,\opW  \}$, \\ \hspace*{5mm} (a) $H \models \opA \Op \phi $ ($H \models \phi \opA \Op \psi$) if for all infinite paths $S$ of $T(G)$, which start in $H$, $S$ satisfies $\Op\phi$ \hspace*{11mm} ($\phi \opA \Op \psi$) in the LTL-sense, e.g.,\ $H \models \opAG\phi$ means: for every infinite path $\tuple{H=H_0  \dder_\R \ldots }$, \hspace*{11mm} for every $i \ge 0$, $H_i \models \phi$. \\\hspace*{6mm} (b) $H \models \opE \Op \phi $ ($H \models \phi \opE \Op \psi$) if there is an infinite path of $T(G)$, which starts in $H$ and satisfies $\Op\phi$ \hspace*{11mm} ($\phi \opE \Op \psi$) in the LTL-sense, e.g.,\ $H \models \opEX \phi$ means: there is an infinite path $\tuple{H=H_0 \dder_\R \ldots }$ \hspace*{11mm} s.t.\ $H_1 \models \phi$.\\
For a CTL constraint $\phi$, $T(G)$ \emph{satisfies} $\phi$, in symbols $T(G) \models \phi$, if $G \models \phi$.\\ A GTS $\R$ \emph{satisfies} $\phi$, in symbols $\R \models \phi$ if $T(G) \models \phi$ for all transformation trees $T(G)$ w.r.t.\ $\overline{\R}$ (\textcolor{new}{i.e.,\ for all graphs $G$}).\end{defi}}

\section{Joint Graph Transformation Systems}\label{joint}
In this section, we \textcolor{new}{define joint graph transformation systems, each of which} involves the system, the environment, and an automaton modeling the interaction between them. Both system and environment are GTSs. We illustrate the way of functioning of a joint graph transformation system by a very simple example, namely the TNS in Example \ref{ex}.

\begin{asm} In the following, let $\Lambda$ be a fixed label alphabet, and $\Sys$ and $\Env$ be GTSs over $\Lambda$. W.l.o.g., we assume that $\Sys$ and $\Env$ are disjoint. (If $\Sys$ and $\Env$ share a common rule $r$, we assign $r$ different names in $\Sys$ and $\Env$.) \end{asm}
We specify the class of automata which are used to regulate the interaction between system and environment. These regulation automata are similar to $\omega$-automata, see, e.g.,\ \cite{Thomas90}.

\begin{defi}[regulation automaton] A \emph{regulation automaton} of $\tuple{\Sys, \Env}$ is a tuple $A=\tuple{Q, q_0, \delta, \select}$ consisting of a finite set $Q$ disjoint from $\Lambda$, called the \emph{state set}, a  \emph{starting state} $q_0 \in Q$, a \emph{transition relation} $\delta \subseteq Q \times Q$, \ignore{whose transitive closure is connected\footnote{A relation $R \subseteq X \times X$ is \emph{connected} iff for all $x,y \in X, x \neq y$: $\tuple{x,y} \in R$ or $\tuple{y,x} \in R$. }} and a function $\select: \delta \to \Pot ( \Sys \cup \Env ) $ (into the power set of $\Sys \cup \Env$), called the \emph{selection function}.  A regulation automaton is \emph{proper} if no element of the transition relation $\delta$ is assigned to the empty set and \textcolor{new}{$\delta^+(q_0 )\supseteq Q\setminus \{q_0\}$, i.e., the image of $\{q_0\}$ under the transitive closure $\delta^+$ contains $Q \setminus \{q_0 \}$}.
\end{defi}
\begin{bem} Alternatively, the interaction between system and environment can be specified by an $\omega$-regular language \cite{Thomas90} over the alphabet $\Pot(\Sys \cup \Env)$ corresponding to the accepting regulation automaton. In order to obtain a GTS as joint system, we prefer to use regulation automata. \end{bem}
An example of a regulation automaton of $\tuple{\Sys, \Env}$ is given in Example \ref{ex} where the system $\Sys$ consists of the rules \texttt{Ascend}, \texttt{Descend}, \texttt{Move}, and \texttt{Repair} and the only environment rule is \texttt{Block}. In the following, we consider only \emph{proper} regulation automata. However, this leads not to a loss of generality (see Proposition \ref{proper}). \par
A joint graph transformation system is obtained by ``synchronizing'' the system, repectively, the environment, with the simulation of the regulation automaton, and then joining both sets of enriched rules. 

\begin{defi}[joint graph transformation system] Let $\Sys$ and $\Env$ be graph transformation systems, called \textit{system} and \textit{environment}, respectively, and $A=\tuple{Q,q_0 , \delta, \select }$ a regulation automaton \textcolor{new}{of $\tuple{\Sys,\Env}$}. The \textit{joint (graph transformation) system} of $\Sys$ and $\Env$ w.r.t. $A$ is the graph transformation system $\Sys_A \cup  \Env_A$, shortly $\Sys\Env$, where for a rule set $\R \in \{\Sys, \Env\}$, $\R_A$ denotes the \textit{enriched rule set} \begin{displaymath} \R_A = \{ \tuple{\tuple{L,q} \dder \tuple{R,q'},\ac} \vert \tuple{L \dder R, \ac} \in \R \cap \select \tuple{q,q'} \text{ and }\tuple{q,q'} \in \delta \} \end{displaymath} \textcolor{new}{where $\tuple{\tuple{L,q} \dder \tuple{R,q'}}$ abbreviates $\brule{\tuple{L,q}}{K}{\tuple{R,q'}}$ for $ \tuple{L \dder R}=\brule{L}{K}{R}$. For a graph $G$ and a state $q$, $\tuple{G,q}$ denotes the graph $G + \raisebox{0.3mm}{\circlearound{$q$}}$, i.e.,\ the disjoint union of $G$ and a node labeled with $q$.}
\end{defi}
\begin{bem} On the one hand, every joint graph transformation system is a graph transformation system over the label set $\Lambda \cup Q$ where $Q$ is the state set of the regulation automaton. On the other hand, for every graph transformation system $\Sys$, one can construct an ``equivalent'' joint graph transformation system using the environment $\Env= \varnothing$ and the regulation automaton given by: \vspace*{3mm} \begin{center}\scalebox{0.9}{\begin{tikzpicture}[-> ,>= stealth' , node distance=3cm]

 \node[state,fill=whitegray] (R){$q_0$};

\node (start) [left = 1cm of R] {};

\path (start) edge node{} (R)
(R) edge[loop above] node{$\mathcal{S}$} (R) ;
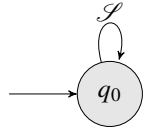
 \end{tikzpicture}}\end{center}  \begin{center} Fig. 4: representation of the considered regulation automaton of $\tuple{\Sys, \varnothing}$\end{center} \end{bem}

\textcolor{new}{ \begin{convention} Joint graph transformation systems are applied to graphs of the form $\tuple{G, q_0} $ where $G$ is a graph over $\Lambda$ and $q_0$ is the starting state of the regulation automaton. \end{convention} The following proposition shows that in our construction of joint graph transformation systems, system rules, and environment rules, respectively, are synchronized with the corresponding transitions of the regulation automaton.}

\begin{propos}[synchronization] \label{para} For a rule $r=\tuple{ L \dder R,\ac} $ of the GTS $\Sys \cup \Env$, let $r_A =\tuple{ \tuple{L,q} \dder \tuple{R,q'}, \ac}  $ be a corresponding enriched rule in the joint system $\Sys \Env $ w.r.t.\ a regulation automaton  $A=\tuple{Q,q_0 , \delta, \select }$. There is a transformation sequence $\tuple{G_0 ,q_0} \dder_{r_{1,A}} \ldots (\dder_{r_{n,A}} \tuple{G_n, q_n})$ in $\Sys \Env$ iff there is a transformation sequence $G_0 \dder_{r_1} \ldots (\dder_{r_n} G_n)$ in $\Sys\cup\Env$ s.t. $\tuple{q_{i-1},q_i} \in \delta$ and $r_i \in \select \tuple{q_{i-1}, q_i}$ for $i \ge 1$. \end{propos}

\begin{proof} By definition of the enriched rule set, for all graphs $G,G'$ over $\Lambda$, $ \tuple{G ,q} \dder_{r_{A}} \tuple{G',q'} \textnormal{  iff  } G \dder_r G' \textnormal{ s.t. } \tuple{q,q'} \in \delta \textnormal{ and } r \in \select \tuple{q,q'}. $ \end{proof}
For every regulation automaton, there is an ``equivalent'' proper one.
\begin{propos}[proper regulation automata]\label{proper} For all graph transformation systems $\Sys$, $\Env$ and every regulation automaton $A$ of $\tuple{\Sys, \Env}$, there is a proper regulation automaton $A'$ s.t.\ every transformation sequence in $\Sys_A \cup \Env_A$ starting from a graph $\tuple{G_0 ,q_0}$ is also a transformation sequence in $\Sys_{A'} \cup \Env_{A'}$, and vice versa. \end{propos}

\begin{proof} \longv{For $A=\tuple{Q,q_0, \delta, \select}$, let $A'=\tuple{Q',q_0, \delta', \select}$ where $Q' \subseteq Q$ is the set of states, which can be reached from $q_0$ without crossing transitions $\tuple{q,q'}$ with $\select\tuple{q,q'}=\varnothing$, and $\delta' = \delta \cap Q' \times Q'$. The inclusion $\Sys_{A'} \cup \Env_{A'} \subseteq \Sys_{A} \cup \Env_{A}$ yields one direction of the statement. Conversely, consider the transformation sequence $\tuple{G_0 ,q_0} \dder_{r_1} \ldots (\dder_{r_n} \tuple{G_n, q_n}) $ in $\Sys_A \cup \Env_A$. Assume that there is $i \ge 1 $ s.t. $r_i = \tuple{\tuple{L, q_{i-1}}  \dder \tuple{ R, q_i} , \ac} \not \in \Sys_{A'} \cup \Env_{A'}$. By construction of $A'$, $q_i$ cannot be reached from $q_0$ without crossing a transition $\tuple{q,q'}$ with $\select\tuple{q,q'}=\varnothing$. This is a contradiction to Proposition \ref{para}.}\napx{\shortv{For $A=\tuple{Q,q_0, \delta, \select}$, let $A'=\tuple{Q',q_0, \delta', \select}$ where $Q' \subseteq Q$ is the set of states, which can be reached from $q_0$ without crossing transitions $\tuple{q,q'}$ with $\select\tuple{q,q'}=\varnothing$, and $\delta' = \delta \cap Q' \times Q'$.}}\apx{We can construct $A'$ explicitly. See Appenix \ref{properA}.} \end{proof}

\begin{example}[traffic network system]
Consider the system $\Sys$, the environment $\Env$, and the regulation automaton $A$ given in Example \ref{ex}. The joint graph transformation system $\Sys_A \cup \Env_A$ is given by: \vspace*{3mm}
\begin{center}
$\Sys_A \left\{ \begin{array}{ll} \texttt{Ascend}_A &: \left\tuple{ \track \hspace*{2mm} \circlearound{$q_0$} \hspace*{2mm}\Rightarrow\hspace*{2mm} \onecar \hspace*{2mm}\circlearound{$q_0$} \right} \\ \texttt{Descend}_A &: \left\tuple{ \onecar  \hspace*{2mm} \circlearound{$q_0$} \hspace*{2mm} \Rightarrow \track  \hspace*{2mm} \circlearound{$q_0$} \right} \\ \texttt{Move}_A &: \left\tuple{ \moveL \hspace*{2mm}\circlearound{$q_0$} \hspace*{2mm}\Rightarrow \hspace*{2mm}\moveR \hspace*{2mm}\circlearound{$q_0$} \right}  \\ \texttt{Repair}_A &: \left\tuple{ \blocked \hspace*{2mm} \circlearound{$q_1$} \hspace*{2mm}\Rightarrow \track \hspace*{2mm} \circlearound{$q_0$} \right} \end{array} \right.$ \\ \vspace*{2mm}\end{center}
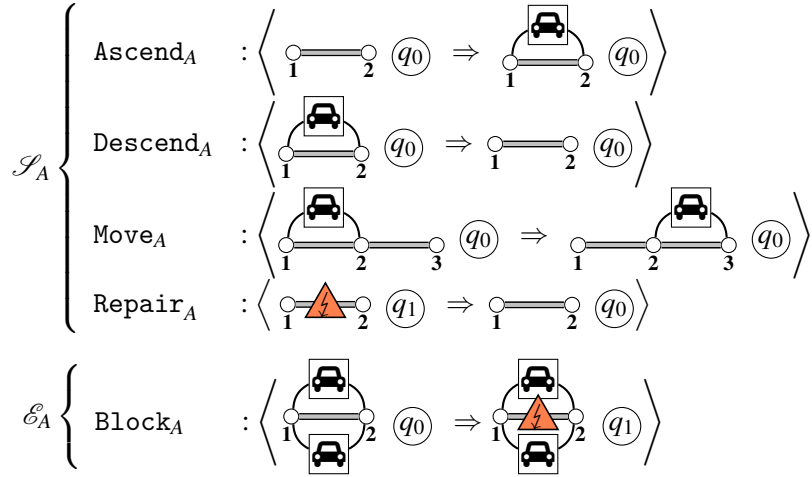
\hspace*{27mm}$\Env_A \left\{  \begin{array}{ll}\texttt{Block}_A&\hspace*{4mm}: \left\tuple{ \twocar \hspace*{2mm} \circlearound{$q_0$} \hspace*{2mm}\Rightarrow \blockedtwo \hspace*{2mm} \circlearound{$q_1$}  \right} \end{array} \right.$  \vspace*{3mm}\begin{center} Fig. 5: the TNS as joint graph transformation system \end{center}
\end{example}

\section{Correctness Notions}
\label{correct}
\shortv{
In this section, we recall a notion of correctness for GTSs and define two instances of correctness notions for joint systems, which generalize correctness for GTSs, namely \emph{k-step} and \emph{last-minute correctness}. }\longv{
In this section, we recall a notion of correctness for GTSs and define three instances of correctness notions for joint systems, which generalize correctness for GTS, namely \emph{k-step}, \emph{weak k-step}, and \emph{last-minute correctness}. } \par
The definition of (partial) correctness of a GTS w.r.t.\ a \emph{precondition} and a \emph{postcondition} means that for every graph satisfying the precondition, every graph derived via the GTS satisfies the postcondition.
This is a stronger notion than the usual definition of correctness since we demand that the postcondition is valid in every following step and not only in the first step.
\begin{defi}[correctness] Let $\tuple{c,d}$ be a pair of graph constraints over $\Lambda$. A graph transformation system $\R$ is \emph{(partially) correct} w.r.t.\ $\tuple{c,d}$ if for all graphs $G,H$ over $\Lambda$ with $G \models c$ and $G \dder_{\R}^+ H$, $H \models d$ holds. We call $c$ the \textit{precondition} and $d$ the \textit{postcondition}.  \end{defi}
\begin{example}[system correctness of the TNS] It is clear that the TNS is not correct\ignore{\footnote{In the context of correctness, enriched rule sets and joint GTSs are treated as GTSs over $\Lambda$.}} w.r.t.\ $\tuple{c,c}$ where $c=\text{NoBlocked}$. However, the postcondition is violated only after applying $\texttt{Block}_A$. \ignore{If we consider only the enriched system rules, we obtain correctness.}The enriched system $\Sys_A$ of the traffic network is correct w.r.t.\ $\tuple{c,c}$ since \texttt{Block} is not a system rule. \end{example}

\begin{convention} In the context of correctness, enriched rule sets and joint GTSs are treated as GTSs over $\Lambda$. \end{convention}
For joint systems, this correctness notion is too restrictive. We allow that the validity of the postcondition is violated for a short period after an interference of the environment. Nonetheless, the postcondition must be recovered. We give two examples of correctness notions. \par One possibility is to limit the maximal number of steps, after which the postcondition must be recovered. 

\begin{asm} In the following, let $\SysEnv$ be a joint system w.r.t.\ a regulation automaton $A=\tuple{Q,q_0, \delta, \select}$ and $\tuple{c,d}$ be a pair of graph constraints over $\Lambda$.  \end{asm}
The notion of $k$-step correctness demands that after an interference of the environment, recovery of the postcondition must occur in at most $k$ steps. 
\begin{defi}[$k$-step correctness] Let $k$ be a natural number. The joint system $\SysEnv$ is $k$-\emph{step correct} w.r.t. $\tuple{c,d}$ if \begin{itemize} \item[(S)] the enriched system $\Sys_A$ is correct w.r.t.\ $\tuple{c ,d}$ and \item[(R$^k$)] for every graph $G=\tuple{G',q_0}$ with $G \models c $ \ignore{\footnote{$G \models c$ iff $G' \models c$ since $c$ is a graph constraint over $\Lambda$ and $G'$ is a graph over $\Lambda$.}} and every transformation sequence $\tuple{G \dder_{\Sys\Env}^* H \dder_{\Env_A} M \dder_{\SysEnvC}^k  N  }$ there is a subsequence $\tuple{M \dder_{\SysEnvC}^{\le k} N'}$ with $N' \models d$ (\emph{k-step recovery}). \end{itemize} \end{defi}

\begin{example}[1-step correctness of the TNS]  Let $c=\text{NoBlocked}$. The traffic network system is neither correct w.r.t.\ $\tuple{c,c}$ nor 0-step correct but 1-step correct w.r.t.\ $\tuple{c,c}$ since after applying $\texttt{Block}_A$, there is a blocked track and immediately, by applying $\texttt{Repair}_A$, the blocked track vanishes. \end{example}
The notion of $k$-step correctness yields a hierarchy. By the following proposition, $k$-step correctness implies $(k+1)$-step correctness for $k \ge 1$. Hence, the TNS is also $k$-step correct w.r.t.\ $\tuple{c,c}$ for $k \ge 1$.

\begin{propos}[hierarchy of $k$-step correctness]\label{hier} Let $k$ be a natural number. \begin{itemize} \item[(i)]  If $\SysEnv$ is correct w.r.t.\ $\tuple{c,d}$, \ignore{\newline} then it is $0$-step correct w.r.t.\ $\tuple{c,d}$.  \item[(ii)]If $\SysEnv$ is $k$-step correct w.r.t.\ $\tuple{c,d}$, \ignore{\newline} then it is $(k+1)$-step correct w.r.t.\ $\tuple{c,d}$. \end{itemize} \end{propos}
\begin{proof} \ignore{(i) Let $G=\tuple{G',q_0}$ with $G'$ a graph over $\Lambda$ and $G \models c $. Consider a transformation sequence $\tuple{G \dder_{\Sys\Env}^* H \dder_{\Env} M}$. The correctness of $\SysEnv$ implies that $M \models d$. \\
(ii)} This follows by definition. \end{proof}\vspace*{4mm}
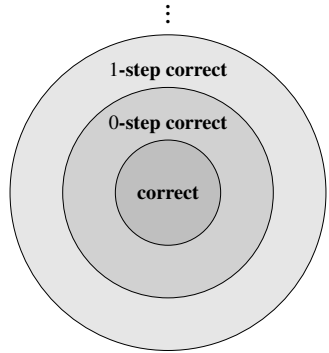
\begin{center}\scalebox{0.7}{\begin{tikzpicture} \draw[fill=whitegray] (0,0) circle [radius=3] node {};\draw[fill=lightergray] (0,0) circle [radius=2] node {}; \draw[fill=lightgray] (0,0) circle [radius=1] node {\textbf{correct}};  \draw (0,1.3) node {\textbf{$0$-step correct}}; \draw (0,2.3) node {\textbf{$1$-step correct}};\draw (0,3.5) node {\textbf{$\vdots$}}; \end{tikzpicture}} \end{center} \begin{center} Fig. 6: illustration of the hierarchy of $k$-step correctness \end{center}
\vspace*{4mm}We introduce another notion of correctness which demands not a recovery after $k$ steps but a recovery until the next enviroment step, i.e.,\ the postcondition must hold at the last point in time before the next environment rule is applied.  \newpage
\begin{defi}[last-minute correctness] A joint graph transformation system $\SysEnv$ is \emph{last-minute correct} w.r.t.\ $\tuple{c,d}$ if \begin{itemize} \item[(S)] the enriched system $\Sys_A$ is correct w.r.t. $\tuple{c ,d}$ and \item[(R)] for every transformation sequence $\tuple{G \Rightarrow_{\SysEnv}^* M \Rightarrow_{\Env_A} N \Rightarrow_{\Sys_A}^+ H \dder_{\Env_A} H'}$ with $G=\tuple{G',q_0}$ and $G \models c $, $H \models d$ holds (\emph{last-minute recovery}). \end{itemize} \end{defi}

\begin{example}[last-minute correctness of the TNS]  Let $c=\text{NoBlocked}$. The traffic network system is last-minute correct w.r.t. $\tuple{c,c}$ since after applying $\texttt{Block}_A$ there is a blocked track but immediately, by applying $\texttt{Repair}_A$, the blocked track vanishes and there is no blocked track until again $\texttt{Block}_A$ is applied. \end{example}
Condition (R) can be replaced by a simpler condition (R') if we assume the correctness of the enriched system. We have the following characterization: 
\begin{propos}[characterization of last-minute correctness] A joint system $\Sys\Env$ is last-minute correct w.r.t $\tuple{ c, d }$ iff the condition (S) holds and (R') for every transformation sequence $\tuple{G \Rightarrow_{\SysEnv}^* N \Rightarrow_{\Sys_A} H \dder_{\Env_A} H'}$ with $G=\tuple{G',q_0}$ and $G \models c $ , $H \models d$ holds. 
 \end{propos}
\begin{proof} ``(R') $\dder$ (R)'': Assume that (R') holds. Consider a transformation sequence $\tuple{G \Rightarrow_{\SysEnv}^* M \Rightarrow_{\Env_A} N \Rightarrow_{\Sys_A}^+ H \dder_\Env H'}$ with $G \models c  $. Then there is a graph $N'$ s.t. $N \dder_{\Sys_A}^* N' \dder_{\Sys_A} H$. We apply condition (R') to the transformation sequence $\tuple{ G \Rightarrow_{\SysEnv}^* N' \Rightarrow_{\Sys_A} H \dder_{\Env_A} H' }$ and obtain $H \models d$. \\
 ``(R) $\land$ (S) $\dder$ (R')'': Assume that (R) and (S) hold. Consider a transformation sequence $ \tuple{G \Rightarrow_{\SysEnv}^* N \Rightarrow_{\Sys_A} H  \dder_{\Env_A} H'}$ with $G \models c  $. If the subsequence $\tuple{G \Rightarrow_{\SysEnv}^* N \Rightarrow_{\Sys_A} H}$ is also transformation sequence in $\Sys_A$, $H \models d$ since the enriched system is correct w.r.t. $\tuple{c,d}$. If this is not the case, there exist graphs $M',N'$ s.t. $G \dder_{\SysEnv}^* M' \dder_{\Env_A} N' \dder_{\Sys_A}^+ H$. We apply (R) to the transformation sequence $\tuple{G \dder_{\SysEnv}^* M' \dder_{\Env_A} N' \dder_{\Sys_A}^+ H \dder_{\Env_A} H'}$ and obtain $H \models d$.  \end{proof}
\apx{
\begin{bem} Last-minute $k$-step correctness are instances of correctness notions for joint GTSs. One may define other notions, e.g. demanding recovery in only one following path. This sort of notions cannot be formalized in LTL but in \emph{computation tree logic} (CTL). We added an instance of correctness which can be formalized in CTL, to Appendix \ref{CTL} including also the proof. \end{bem}}
\napx{\shortv{\begin{bem}These notions are instances of correctness notions for joint GTSs. One may define other notions, which can be formalized not in LTL but e.g. in CTL. \end{bem}}}
\shortv{We close this section with an overview of the correctness notions:}\longv{We give an overview of the correctness notions:}
\begin{center} \tikzset{
    state/.style={
           rectangle,
           rounded corners,
           draw=black,
           minimum height=2em,
           inner sep=2pt,
           text centered,
           },
}
\vspace*{6mm} \begin{tikzpicture}
\node[state,fill=whitegray] (corr) at (0,0) {\hspace*{8.5mm}correctness of $\SysEnv$\hspace*{8.5mm}};
\node[state,fill=whitegray] (corrS) at (6,-1.5) {\hspace*{8.8mm}correctness of $\Sys_A$\hspace*{8.8mm}};
\node[state,fill=whitegray] (lmcorr) at (6,0) {last-minute correctness of $\SysEnv$};
\node (dot) at (0,-4.35) {$\vdots$};
\node[state,fill=whitegray] (zerocorr) at (0,-1.5) {\hspace*{3mm}$0$-step correctness of $\SysEnv$\hspace*{3mm}};

\draw[-implies,double equal sign distance] (corr) -- (lmcorr);  
\draw[-implies,double equal sign distance] (corr) -- (zerocorr);  

\draw[-implies,double equal sign distance] (2,-2.8) -- (3.75,-1.85);  
\node[state,fill=whitegray] (onecorr) at (0,-3) {\hspace*{3mm}$1$-step correctness of $\SysEnv$\hspace*{3mm}};
\draw[-implies,double equal sign distance] (zerocorr) -- (onecorr);  
\draw[-implies,double equal sign distance] (zerocorr) -- (corrS);  

\draw[-implies,double equal sign distance] (lmcorr) -- (corrS);  

\draw[-implies,double equal sign distance] (onecorr) -- (dot);  
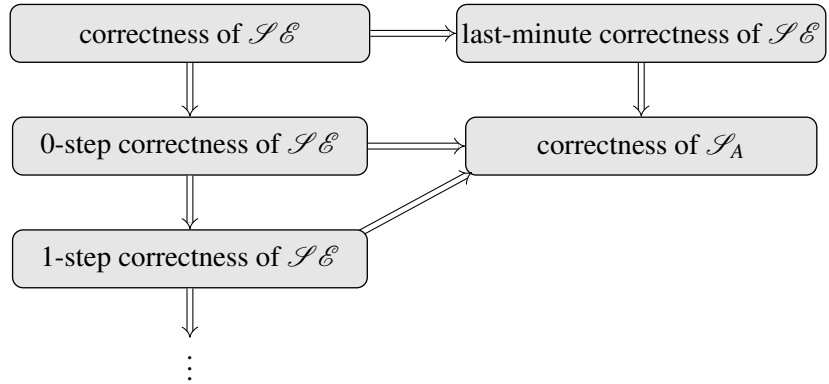
\end{tikzpicture} \end{center}\begin{center} Fig. 7: relations of the correctness notions\end{center}  \vspace*{3mm} 
\longv{The notion of weak $k$-step correctness demands that after an interference of the environment, in at least one following sequence, recovery of the postcondition must occur in at most $k$ steps. \newpage
\begin{defi}[weak $k$-step correctness] Let $k$ be a natural number. A joint system
$\SysEnv$ is \emph{weakly k-step correct} w.r.t. $\tuple{c,d}$ if \begin{itemize} \item[(S)] the enriched system GTS $\Sys$ is correct w.r.t.\ $\tuple{c,d}$ and \item[(R$^k_{\text{w}}$)] for every graph $G=\tuple{G',q_0 }$ with $G'$ a graph over $\Lambda$, $G \models c$ and every transformation sequence $\tuple{G \dder_{\Sys\Env}^* H \dder_{\Env_A} M}$ there is an $N$ s.t. $M \dder_{\SysEnv}^{\le k} N$ and $N \models d$ (\emph{weak k-step recovery}). \end{itemize} \end{defi}
Weak $k$-step correctness is, indeed, a weaker notion than $k$-step correctness. Similarly to $k$-step correctness, we obtain a hierarchy for weak $k$-step correctness.
\begin{propos}[relation to \emph{k}-step correctness]\label{weakhier} Let $k$ be a natural number. \begin{itemize} \item[(i)] If $\SysEnv$ is $k$-step correct w.r.t.\ $\tuple{c,d}$, then $\SysEnv$ is weakly $k$-step correct w.r.t. $\tuple{c,d}$.    \item[(ii)] If $\SysEnv$ is weakly $k$-step correct w.r.t.\ $\tuple{c,d}$, then $\SysEnv$ is weakly $(k+1)$-step correct w.r.t. $\tuple{c,d}$. \end{itemize} \end{propos}

\begin{proof} This follows by definition. \end{proof}}
The following example is weakly $k$-step correct for every $k \ge 1$ but not $k$-step correct for any $k$ w.r.t.\ the same pre- and postcondition.
\begin{example}[weak 1-step correctness] Let $c=\text{NoBlocked}$. Consider the system and environment in Example \ref{ex}. Let $B$ be the regulation automaton given by the following figure: \vspace*{3mm}\begin{center}\scalebox{0.9}{\begin{tikzpicture}[-> ,>= stealth' , node distance=3cm]

 \node[state,fill=whitegray] (R){$q_0$};
 \node[state, fill=whitegray] (S) [right of=R]{$q_1$};

\node (start) [left = 1cm of R] {};

\path (start) edge node{} (R) 
(R) edge[bend left,above] node{\texttt{Block}} (S)
(R) edge[loop above, above] node{\small{\{\texttt{Move}, \texttt{Ascend}, \texttt{Descend}\}}} (S)
(S) edge[loop above, above] node{\small{$\Sys$}} (S)
 (S) edge[bend left, below] node {\texttt{Repair}} (R) ;
 \end{tikzpicture}} \end{center}  \begin{center} Fig. 8: representation of the regulation automaton $B$ \end{center} \vspace*{3mm} In this case, the joint system $\Sys_B \cup \Env_B$ is weakly 1-step w.r.t.\ $\tuple{c,c}$ since $\texttt{Repair}_B$ can be applied after an application of $\texttt{Block}_B$ and before $\texttt{Block}_B$ is applied again, the blocked track must vanish due to an application of $\texttt{Repair}_B$. (By Proposition \ref{weakhier}, it is weakly $k$-step correct for $k \ge 1$.) It is not $k$-step correct w.r.t.\ $\tuple{c,c}$ for any $k$ since the other system rules may be applied arbitrarily often after an application of $\texttt{Block}_B$. \end{example}
\ignore{\begin{example}[weak 1-step correctness] Let $c=\text{NoBlocked}$. Similar to the proof of Proposition \ref{applic}, using the constructions in \cite{Pennemann06}, this graph constraint can be transformed into an application condition. For the rule $\texttt{Block}$, the application condition $\text{acNoBlocked}$ states that there is no blocked track attached to the left handside \twocarNoIndex of $\texttt{Block}$ or outside of the left handside.   
The unregulated traffic network system, considered as a joint system w.r.t.\ the trivial regulation automaton $U$ ($Q=\{q_0 \}$ and $\select \tuple{q_0,q_0}=\Sys \cup \Env$), is weakly 1-step correct w.r.t.\ $\tuple{c,c}$ since $\texttt{Repair}_U$ can always be applied if there is a blocked track. It is not $k$-step correct w.r.t.\ $\tuple{c,c}$ for any $k$. \end{example}}

\section{Reduction to Model Checking}
\label{reduct}

In this section, we show as main result that our correctness notions can be expressed as \textcolor{new}{LTL or CTL constraints, respectively, i.e.,\ that verification can be reduced to LTL or CTL model checking, respectively}. This is favorable since the tool GROOVE \cite{GROOVE} provides a way for LTL\longv{ and CTL} model checking for systems where the states constitute graphs. \par
First, we refine our notion of joint graph transformation systems, namely to \emph{annotated joint graph transformation systems} which also carry the information whether the last applied rule was a system or environment rule. This is realized by a node labeled with ``$\ssys$'' or ``$\senv$''. 
\begin{notation} For a joint graph transformation system $\SysEnv$, the symbol $\m(\Sys):=\ssys$, $\m(\Env):=\senv$, respectively, is the \emph{marking} of $\Sys$, $\Env$, respectively. For a rule $r \in \R$ and $\R \in \{\Sys, \Env\}$, let $\m(r):=\m(\R)$ be the marking of $r$. Further, we define $\prm(r):=\cup_{q \in Q_{\text{pre}}(r) } \{\m(r') \vert r'\in \Sys \cup \Env :q \in Q_{\text{post}}(r')\} \cup \{ \top \vert q=q_0 \} $ \linebreak $\subseteq \{\top, \ssys, \senv\}$ as \emph{premarkings} of $r$ where $Q_{\text{pre}}(r):= \{q \in Q \vert \exists q' \in Q: \tuple{q,q'} \in \delta, r \in \select\tuple{q,q'}\}$ and $Q_{\text{post}}(r'):= \{q \in Q \vert \exists q'' \in Q: \tuple{q'',q} \in \delta, r' \in \select\tuple{q'',q}\}$.\end{notation}
The set of premarkings of $r$ depicts all markings of rules which could have been applied prior to an application of $r$. In the following definition, it can simply be thought of $ \{\top, \ssys, \senv\}$ (the set of all possible markings) instead of the premarkings.   

\begin{defi}[annotated joint graph transformation system] Let $\Sys\Env$ be a joint graph transformation systems w.r.t.\ a regulation automaton $A=\tuple{Q,q_0 , \delta, \select }$ of $\tuple{\Sys,\Env}$. The \textit{annotated joint (graph transformation) system} of $\Sys$ and $\Env$ w.r.t.\ $A$ is ${\Sys'}_A \cup  {\Env'}_A$, shortly $(\Sys \Env)'$, where for a rule set $\R \in \{\Sys,\Env\}$, ${\R'}_A$ denotes the \textit{marked rule set} \begin{align*} \R'_A := \{ &\tuple{\tuple{L,q,m} \dder \tuple{R,q',m'},\ac} \vert  \\ &\tuple{\tuple{L,q} \dder \tuple{R,q'},\ac} \in \R_A, m \in \prm\tuple{L \dder R, \ac}, m' =\m( \R )\} &\end{align*} \textcolor{new}{where $\tuple{\tuple{L,q,m} \dder \tuple{R,q',m'}}$ abbreviates $\brule{\tuple{L,q,m}}{K}{\tuple{R,q',m'}}$ for $ \tuple{L \dder R}=\brule{L}{K}{R}$. For a graph $G$, a state $q$, and a marking $m$, $\tuple{G,q,m}$ denotes the graph $G + \raisebox{0.3mm}{\circlearound{$q$}}+\circlearound{$m$}$, i.e.,\ the disjoint union of $G$, a node labeled with $q$, and a node labeled with $m$.}
\end{defi}\apx{An example of an annotated joint GTS can be found in Appendix \ref{annotated}.}\napx{
\longv{\begin{example}[annotated TNS] Consider the system $\Sys$, the environment $\Env$, and the regulation automaton $A$ given in Example \ref{ex}. The annotated joint system $\Sys'_A \cup \Env'_A$ is given by: \vspace*{3mm}
\begin{center}
$\Sys'_A \left\{ \begin{array}{ll} \texttt{Ascend}'_{A,\top} &: \left\tuple{ \track \hspace*{2mm} \circlearound{$q_0$} \hspace*{2mm} \ntop \hspace*{2mm}\Rightarrow\hspace*{2mm} \onecar \hspace*{2mm}\circlearound{$q_0$} \hspace*{2mm} \off \right} \\  \texttt{Ascend}'_{A,\ssys} &: \left\tuple{ \track \hspace*{2mm} \circlearound{$q_0$}  \hspace*{2mm} \off\hspace*{2mm}\Rightarrow\hspace*{2mm} \onecar \hspace*{2mm}\circlearound{$q_0$}  \hspace*{2mm} \off\right} \\ \texttt{Descend}'_{A,\top} &: \left\tuple{ \onecar  \hspace*{2mm} \circlearound{$q_0$} \hspace*{2mm} \ntop \hspace*{2mm}  \Rightarrow \track  \hspace*{2mm} \circlearound{$q_0$} \hspace*{2mm} \off \right} \\ \texttt{Descend}'_{A,\ssys} &: \left\tuple{ \onecar  \hspace*{2mm} \circlearound{$q_0$} \hspace*{2mm} \off \hspace*{2mm}  \Rightarrow \track  \hspace*{2mm} \circlearound{$q_0$} \hspace*{2mm} \off \right} \\ \texttt{Move}'_{A,\top} &: \left\tuple{ \moveL \hspace*{2mm}\circlearound{$q_0$} \hspace*{2mm} \ntop \hspace*{2mm}\Rightarrow \hspace*{2mm}\moveR \hspace*{2mm}\circlearound{$q_0$}\hspace*{2mm} \off \right}  \\  \texttt{Move}'_{A,\ssys} &: \left\tuple{ \moveL \hspace*{2mm}\circlearound{$q_0$} \hspace*{2mm} \off \hspace*{2mm}\Rightarrow \hspace*{2mm}\moveR \hspace*{2mm}\circlearound{$q_0$}\hspace*{2mm} \off \right}  \\ \texttt{Repair}'_A &: \left\tuple{ \blocked \hspace*{2mm} \circlearound{$q_1$} \hspace*{2mm} \on \hspace*{2mm}\Rightarrow \track \hspace*{2mm} \circlearound{$q_0$}\hspace*{2mm} \off \right} \end{array} \right.$ \\ \vspace*{2mm}\end{center}
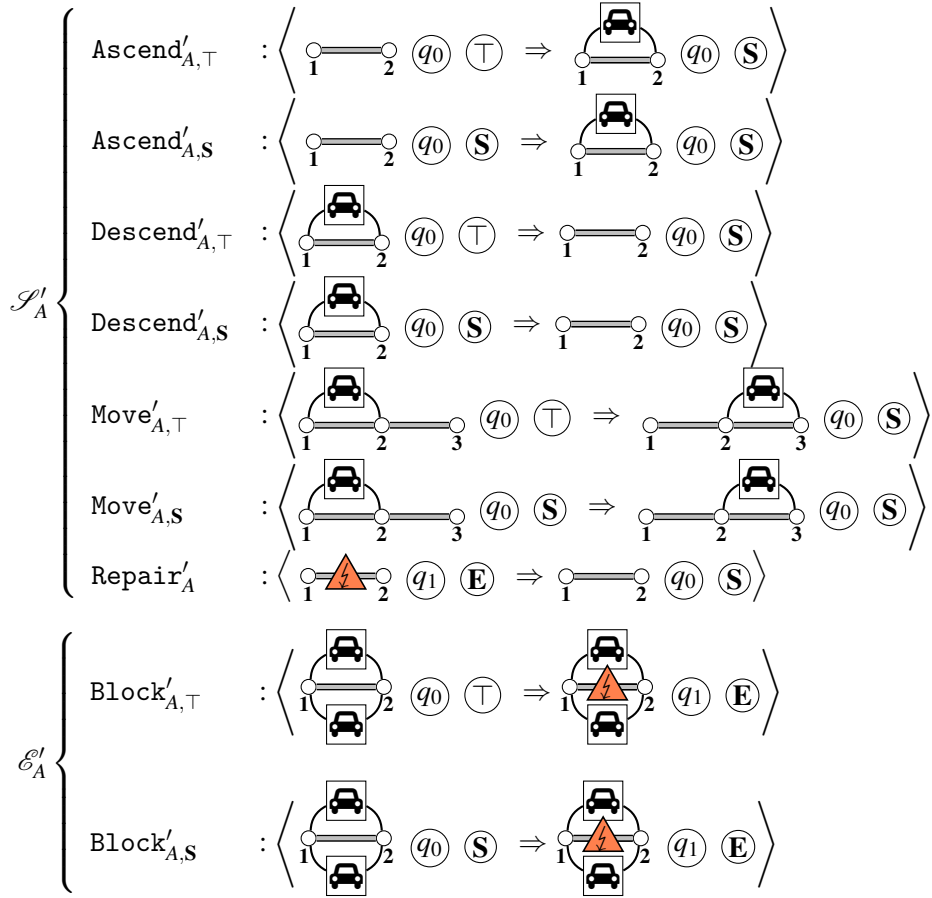
\hspace*{18.5mm}$\Env'_A \left\{  \begin{array}{ll} \texttt{Block}'_{A,\top}&\hspace*{4mm}: \left\tuple{ \twocar \hspace*{2mm} \circlearound{$q_0$} \hspace*{2mm} \ntop\hspace*{2mm}\Rightarrow \blockedtwo \hspace*{2mm} \circlearound{$q_1$} \hspace*{2mm} \on \right} \\ & \\ \texttt{Block}'_{A,\ssys}&\hspace*{4mm}: \left\tuple{ \twocar \hspace*{2mm} \circlearound{$q_0$} \hspace*{2mm} \off\hspace*{2mm}\Rightarrow \blockedtwo \hspace*{2mm} \circlearound{$q_1$} \hspace*{2mm} \on \right}  \end{array} \right.$ \begin{center} \vspace*{3mm}Fig. 9: the TNS as annotated joint system\end{center}\end{example}}}
\vspace*{3mm}For every joint system, there is an ``equivalent'' annotated joint system.
\begin{propos}[annotation] For an enriched rule $r=\tuple{ \tuple{L,q} \dder \tuple{R,q'}, \ac} $ of a joint graph transformation system $\Sys \Env$ w.r.t. a regulation automaton  $A=\tuple{Q,q_0 , \delta, \select }$, let $r'=\tuple{ \tuple{L,q,m} \dder \tuple{R,q',m'}, \ac}  $ be a corresponding marked rule in the annotated joint graph transformation system $(\Sys \Env)'$. There is a transformation sequence $\tuple{G_0 ,q_0} \dder_{r_1} \ldots (\dder_{r_n} \tuple{G_n, q_n})$ in $\Sys \Env$ iff there is a transformation sequence $\tuple{G_0 ,q_0, \top} \dder_{r'_1}\tuple{G_1, q_1, m_1 } \dder_{r'_2}\ldots (\dder_{r'_n} \tuple{G_n, q_n, m_n})$ in $(\Sys\Env)'$ with markings $m_i \in \{\ssys, \senv\}$ for $i \ge 1$. Moreover, in this case, for all $i \ge 1$, $m_i = \m(\R)$ iff $r_{i} \in \R_A$. \end{propos}

\begin{proof} For every $G, G'$ over $\Lambda$, $\tuple{G,q} \dder_r \tuple{G',q'} \textnormal{ iff } \tuple{G,q,m} \dder_{r'} \tuple{G',q',m'}$. By definition of marked rule set, $m' = \m(\R)$ iff $r \in \R_A$. \end{proof}

\begin{bem} Annotated joint graph transformation systems are graph transformation systems over the label set $\Lambda \cup Q \cup \{\top, \ssys, \senv\}$.  \end{bem}
\textcolor{new}{\begin{convention} Annotated joint graph transformation systems are applied to graphs of the form $\tuple{G, q_0, \top} $ where $G$ is a graph over $\Lambda$ and $q_0$ is the starting state of the regulation automaton.\end{convention}}
\longv{\subsection{Reduction LTL Model Checking}
We first consider the formalization of $k$-step and last-minute correctness as LTL constraints.} The following theorem states that checking $k$-step/last-minute correctness of a joint system is equivalent to checking whether the annotated joint system satisfies a certain LTL constraint. 

\begin{satz}[from correctness to LTL] \label{main}For every pair $\tuple{c,d}$ of graph constraints and every natural number $k$, there are LTL constraints $\phi\tuple{c,d}$ and $\phi_k\tuple{c,d}$ s.t.\ for every joint system $\Sys\Env$: \begin{itemize} \item[(a)] $\SysEnv$ is $k$-step correct w.r.t.\ $\tuple{c,d}$ iff $(\SysEnv)' \models \phi_k \tuple{c,d}$. \item[(b)] $\SysEnv$ is last-minute correct w.r.t.\ $\tuple{c,d}$ iff $(\SysEnv)' \models \phi \tuple{c,d}$. \end{itemize}  \end{satz}
For the condition of system correctness, the $k$-step, and the last-minute recovery condition, we construct LTL constraints $\PCS \tuple{c,d}$, $\kSC\tuple{c,d}$, and $\GR\tuple{c,d}$, respectively, where $\tuple{c,d}$ are the pre- and postcondition. Each of the considered LTL constraint formalizes the corresponding condition. By combining them, we obtain conjunctions each formalizing the corresponding correctness notion.
\begin{construction}  Let \\
\begin{minipage}[h]{0.47\textwidth}
\begin{center}\begin{align*}
\phi_k\tuple{c,d} &:=  \PCS\tuple{c,d} \land  \kSC\tuple{c,d} \\
\phi \tuple{c,d}&:=  \PCS\tuple{c,d} \land  \GR\tuple{c,d} \\
\end{align*} \end{center}
\end{minipage}\begin{minipage}[h]{0.05cm}\hfill \rule[-8mm]{0.3mm}{18mm} \hfill\end{minipage}
\begin{minipage}[h]{0.47\textwidth}
\begin{center}\begin{align*}  \PCS\tuple{c,d} &:= c \dder \opX( (\sys \land d ) \opW \env ) \\
\kSC \tuple{c,d} & := c  \dder \opX \opG ( \env \dder \vee_{j=0}^k \opX^j d ) \\
\GR\tuple{c,d}  &:= c  \dder \opG ( s \land \opX e \dder d ) \\
 \end{align*}\end{center} \end{minipage} \\
where $k$ is a natural number, $\opX^0 $ is ``an empty operator'', and for $j \ge 0$, $\opX^{j+1}:=~ \opX^j \opX$ denotes the iterated next-operator in LTL. The graph constraint $s:=\exists (\circlearound{$\ssys$})$ means that the last applied rule was a system rule while $e:=\exists (\circlearound{$\senv$})$ means that the last applied rule was an environment rule. \end{construction}
We split the proof into three lemmata showing that the formalizations of correctness notions as LTL formulas are valid. We start with the formalization of the correctness of the enriched system. \par
\textcolor{new}{The LTL constraint $\PCS\tuple{c,d}$ can be read as ``If $c$ holds at the start, from the following step on, either a system rule was applied ($s$ holds) and $d$ holds forever or until an environment rule was applied ($e$ holds)''.}
\begin{lemma}[from system correctness to LTL] \label{lemS}For a joint graph transformation system $\SysEnv$, the enriched system $\Sys_A$ is correct w.r.t.\ $\tuple{c,d}$ iff $(\SysEnv)' \models \PCS\tuple{c,d}$. \end{lemma}

\begin{proof} ``$\dder$'': Let $S= \tuple{ G_0 \dder G_1 \dder \ldots} $ be an infinite transformation sequence in $\overline{(\SysEnv)'}$. Assume that $G_0 \models c$. Let $S' \subseteq S$ be the longest connected subsequence of $S$, which is also a sequence in $\Sys'_A$ and starts with $G_0$. Since $\Sys_A$ is correct w.r.t. $\tuple{c ,d}$, every graph $G_i$, $i \ge 1$ in $S'$ satisfies $d$. If $S'=S$, we have that $G_0 \models \opX \opG (\sys \land d )$, which implies $G_0 \models \opX( (\sys \land d ) \opW \env) $. If $S' \neq S$, there is a smallest $j$ with $G_j \models \env$. Due to the correctness of $\Sys_A$, every $G_i$ for $1 \le i \le j-1$ satisfies $d$. In this case, $G_0 \models \opX( (\sys \land d ) \opU \env )$, this implies $G_0 \models \opX( (\sys \land d ) \opW \env )$. \\ We showed $G_0 \models c \dder  \opX( (\sys \land d ) \opW \env )$, so $S \models \PCS\tuple{c,d}$.  \\
``$\Leftarrow$'': Assume that $G \models c $ and $G \dder_{\Sys'_A}^+ H$. We define $S':=\tuple{G  \dder_{\Sys'_A}^+ H}$ which is a finite sequence in $\Sys'_A$. We can complete $S'$ to an infinite sequence $S \supseteq S'$ in $\overline{(\SysEnv)'}$. Now holds $S \models  \PCS\tuple{c,d}$, so $G \models \opX( (\sys \land d ) \opW \env )$. This implies that $H \models d $. \end{proof}
The next lemma shows that the formalization of the recovery condition of $k$-step correctness is valid. Together with Lemma \ref{lemS}, this implies part (a) of Theorem \ref{main}. \par
\textcolor{new}{The LTL constraint $\kSC\tuple{c,d}$ can be read as ``If $c$ holds at the start, from the following step on, whenever an environment rule was applied ($e$ holds), there is a $0\le j \le k$ s.t.\ in $j$ steps $d$ holds''.}
\begin{lemma}[from \emph{k}-step recovery to LTL] \label{lemk}A joint graph transformation system $\SysEnv$ is $k$-step correct w.r.t.\ $\tuple{c,d}$ iff the enriched system $\Sys_A$ is correct w.r.t. $\tuple{c,d}$ and $(\SysEnv)' \models \kSC\tuple{c,d}$. \end{lemma}

\begin{proof} We have to show that (R$^k$) is equivalent to $(\SysEnv)' \models \kSC\tuple{c,d}$. \\ 
``$\dder$'': Let $S= \tuple{ G_0 \dder G_1 \dder \ldots} $ be an infinite transformation sequence in $\overline{(\SysEnv)'}$. Assume that $G_0 \models c$. Let $M= G_i$ with $i \ge 1$ and $M \models \env$ be an element of $S$. We can construct a  subsequence $S':=\tuple{G \dder_{(\Sys\Env)'}^* H \dder_{\Env'_A} M \dder_{\overline{(\SysEnv)'}}^k  N  } \subset S$ since $M \models \env$. By definition of $k$-step correctness, there exists a subsequence $\tuple{M \dder_{\overline{(\SysEnv)'}}^{\le k}  N'} \subset S'$ with $N' \models d$. This means, there is a $0 \le j \le k$ with $M \models \opX^j d $, so $M \models  \vee_{j=0}^k \opX^j d$. We showed $G_0 \models \opX \opG ( \env \dder \vee_{j=0}^k \opX^j d )$. \\
``$\Leftarrow$'': Let $S':=\tuple{G \dder_{(\Sys\Env)'}^* H \dder_{\Env'_A} M \dder_{\overline{(\SysEnv)'}}^k  N  }$ be transformation sequence with $G \models c$. We can complete $S'$ to an infinite transformation sequence $S \supset S'$. We have that $G \models \kSC\tuple{c,d}$, this means $G \models  \opX \opG ( \env \dder \bigvee_{j=0}^k \opX^j d )$. Thus, there exists a $0 \le j \le k$ with $M \models \opX^j d $.
\end{proof}
The next lemma shows that the formalization of the recovery condition of last minute correctness is valid. Together with Lemma \ref{lemS}, this implies part (b) of Theorem \ref{main}.\par
\textcolor{new}{The LTL constraint $\GR\tuple{c,d}$ can be read as ``If $c$ holds at the start, whenever a system rule was applied and an environment rule is applied next, $d$ holds''.}

\begin{lemma}[from last-minute recovery to LTL] \label{lemlm}A joint graph transformation system $\SysEnv$ is last-minute correct w.r.t.\ $\tuple{c,d}$ iff the enriched system $\Sys_A$ is correct w.r.t. $\tuple{c,d}$ and $(\SysEnv)' \models \GR\tuple{c,d}$. \end{lemma}

\begin{proof} ``$\dder$'': Let $S= \tuple{ G_0 \dder G_1 \dder \ldots} $ be a transformation sequence in $\overline{(\Sys\Env)'}$. Assume that $G_0 \models c$. We have to show that for every $i \ge 0$, $G_i \models \sys \land \opX \env \dder  d$. Consider the transformation sequence $\tuple{G_0 \dder \ldots \dder G_i \dder G_{i+1}} \subseteq S$ and assume $G_i \models \sys$, $G_{i+1} \models \env$. Since (R') holds, we obtain that $G_{i} \models d$. \\
``$\Leftarrow$'': Let $S'=\tuple{G \Rightarrow_{(\SysEnv)'}^* N \Rightarrow_{\Sys'_A} H }$ be a transformation sequence with $G \models ~c  $  and $H \models \env$. We can complete $S'$ to a transformation sequence $S \supseteq S'$ in $\overline{(\Sys\Env)'}$. It holds $S \models \GR\tuple{c,d}$, i.e.,\ $S \models c \dder   \opG ( \sys \land \opX \env \dder d)$. Since $G \models c$, also $G \models ~{ \opG ( \sys \land \opX \env \dder  d)}$. Since $N \dder_{\Sys_A'} H$ implies that $N \models \sys$, we obtain $H \models d$. \end{proof}

\begin{proof}[of Theorem 1] This follows directly by Lemma \ref{lemS}, \ref{lemk}, and \ref{lemlm}.\end{proof}
\longv{\subsection{Reduction to CTL Model Checking}
Similar to Theorem \ref{main}, we prove that weak $k$-step correctness can be formalized as CTL constraint. The following theorem states that checking weak $k$-step correctness of a joint system is equivalent to checking whether the annotated joint system satisfies a certain CTL constraint. 

\begin{satz}[from correctness to CTL] For every pair $\tuple{c,d}$ of graph constraints and every natural number $k$, there is a CTL constraint $\theta_k$ s.t.\ for every joint GTS $\SysEnv$, $\SysEnv$ is weakly $k$-step correct iff $(\SysEnv)' \models \theta_k \tuple{c,d}$. \end{satz}
\begin{construction} Let 
\begin{align*} \theta_k\tuple{c,d} &:= \PCSctl \tuple{c,d} \land \kWC\tuple{c,d}\\
\PCSctl\tuple{c,d} &:=  c \dder \opAX( (\sys \land d ) \opAW \env ) \\
\kWC\tuple{c,d}  &:= c  \dder \opAX \opAG ( \env \dder \vee_{j=0}^k \opEX^j d ) 
\end{align*} where $k$ is a natural number, $\opEX^0 $ is ``an empty operator'', and for $j \ge 0$, $\opEX^{j+1}:= \opEX^j \opEX$ denotes the iterated existential next-operator in CTL.\end{construction}
\textcolor{new}{The CTL constraint $\PCSctl\tuple{c,d}$ can be read as ``If $c$ holds at the start, for all following sequences, from the following step on, for all following sequences, either a system rule was applied and $d$ holds forever or until an environment rule was applied''.} It is the ``CTL version'' of the LTL constraint $\PCS\tuple{c,d}$.
\begin{lemma}[from system correctness to CTL] \label{lemIV}For a joint system $\SysEnv$, the enriched system $\Sys_A$ is correct w.r.t. $\tuple{c,d}$ iff $(\SysEnv)' \models \PCSctl\tuple{c,d}$. \end{lemma}

\begin{proof} Similar to the proof of Lemma \ref{lemS}.\end{proof}
\textcolor{new}{The CTL constraint $\kWC\tuple{c,d}$ can be read as ``If $c$ holds at the start, for all following sequences, from the following step on, for all following sequences, whenever an environment rule was applied, there exists a following sequence and $0 \le j \le k$ s.t.\ $d$ holds in $j$ steps''.}

\begin{lemma}[from weak \emph{k}-step recovery to CTL] \label{lemV}A joint system $\SysEnv$ is weakly $k$-step correct w.r.t. $\tuple{c,d}$ iff the enriched system $\Sys_A$ is correct w.r.t.\ $\tuple{c,d}$ and $(\SysEnv)' \models \kWC\tuple{c,d}$. \end{lemma}

\begin{proof} We have to show that (R$^k_{\text{w}}$) is equivalent to $(\SysEnv)' \models \kWC\tuple{c,d}$. \\ 
``$\dder$'': Let $T(G)$ be the transformation tree of $G$ in $\overline{(\SysEnv)'}$. Assume that $G \models c$. We show that for every infinite path $\tuple{G=G_0 \dder G_1 \dder \ldots }$ starting from $G$, holds $G_1 \models \opAG(\env \dder\bigvee_{j=0}^k \opEX^j d  )$, i.e.,\ for every infinite path $\tuple{G=G_0 \dder G_1 \dder \ldots }$ starting from $G$, $G_i \models \env \dder\bigvee_{j=0}^k \opEX^j d$ for all $i \ge 1$. Let $M= G_i$ with $i \ge 1$ and $M \models \env$ be an element of an infinite path $\tuple{G=G_0 \dder G_1 \dder \ldots }$. Since $M \models \env$, we can construct a subsequence $S:=\tuple{G \dder_{(\Sys\Env)'}^* H \dder_{\Env'_A} M}$. By definition of weak $k$-step correctness, there is an $N$ with $N \models d$ and $M \dder_{\overline{(\SysEnv)'}}^{\le k}  N$. This means, there is a $0 \le j \le k$ with $M \models \opEX^j d $, so $M \models  \bigvee_{j=0}^k \opEX^j d$. We showed $G \models \opAX \opAG ( \env \dder \bigvee_{j=0}^k \opEX^j d )$. \\
``$\Leftarrow$'': Let $\tuple{G \dder_{(\Sys\Env)'}^* H \dder_{\Env'_A} M  }$ be transformation sequence with $G \models c$. We have that $G \models \kWC\tuple{c,d}$, this means $G \models \opAX \opAG ( \env \dder \bigvee_{j=0}^k \opEX^j d )$. Thus, there exists a $0 \le j \le k$ with $M \models \opEX^j d $, i.e.,\ there is an $N$ with $N \models d$ and $M \dder_{\overline{(\SysEnv)'}}^{\le k}  N$.
\end{proof}

\begin{proof}[of Theorem 2] This follows directly by Lemma \ref{lemIV} and \ref{lemV}. \end{proof}}

\section{Related Concepts} \label{related}
The basic concept of graph transformation (systems) serving as theoretical foundation for our approach is presented in \cite{Ehrig06}. \\ 
We mention some approaches for \emph{graph-transformational interacting systems}, of which the concept of adverse conditions are a special case: \par
In \cite{Corradini09}, \emph{graph transformation systems with dependencies (d-GTSs)} are presented to model reactive systems. In this setting, the system has to react to stimuli from the environment. Reactions are triggered by the presence of \emph{signals} which are subgraphs in the derived graphs. To this aim, transformation rules are equipped with a \emph{dependency relation} which describes extra relationships between the deleted and created elements. To pave the way for implementation, \emph{transactional graph transformation systems (T-GTSs)} are used. Both d-GTSs and T-GTSs are extensions of GTSs, whereas in our approach, the main constructs are GTSs. \par
In \cite{Taentzer99}, the concept of \emph{distributed graph transformation}, which is using graph transformation as underlying formal framework, is presented and applied to distributed systems. In contrast to our approach, the considered \emph{distributed graphs} comprise two abstraction levels: the network and the local level. On the network level, the topological structure of a distributed system is specified. The local level contains the description of local objects. A \emph{distributed graph rule} consists of a \emph{network rule} and \emph{local rules}\ignore{ for all network nodes which are preserved by the network rule. All newly created network nodes as well as those which should be deleted are equipped with a local graph}. \par
In \cite{Wang06}, a graph-transformational approach for modeling reorganization in multi-agent systems is introduced. Similarly to \cite{Taentzer99}, the considered \emph{multi-graphs} comprise three levels: on the top level, \emph{role graphs} composed of \emph{roles} and their interrelations; on the bottom level, \emph{agent graphs} consisting of \emph{agents} and their interrelations; on the mid level, \emph{connection graphs} having roles and agents as nodes, and the role competences for agents as their directed edges. A \emph{mulit-level graph transformation rule} consists of three graph transformation rules each for one of the three levels. The semantics are specified by \emph{graph class expressions} on each of the three levels.  \\
In \cite{Kreowski13}, \emph{graph-transformational swarms} are presented. A graph-transformational swarm consists of \emph{members} which interact simultaneously in an environment represented by a graph. The members are all of the same \emph{kind} or of different kinds. Kinds and members are modeled as \emph{graph transformational units} each consisting of a set of graph transformation rules. \emph{Control conditions} regulate the application of rules; the regulation automaton plays a similar role in our approach. Graph class expressions specify the initial and terminal graphs whereas we use explicit graph constraints for the pre-/postcondition. \par
In \cite{Kreowski06}, the concept of \emph{autonomous units}, which are a closely related notion to graph transformational units \cite{Kreowski13}, is introduced to model distributed systems composed of highly self-controlled components. A class of control conditions may reduce the degree of nondeterminism of rule applications. Moreover, graph class expressions are used to specify sets of initial and terminal graphs. \par
All mentioned concepts of graph-transformational interacting systems are based on graph transformation as underlying framework. By contrast, the main constructs of our approach are basic graph transformation systems. \par
A notion of correctness for \emph{graph programs} is considered in \cite{Pennemann06,Plump13}.\ignore{ Graph programs are programs whose basic components are  sets of rules, i.e.,\ graph transformation systems, and which are characterized by their closedness under non-deterministic composition, sequential composition, and iteration (as-long-as-possible).} In contrast to our definition, the correctness of graph transformation systems is usually interpreted as its correctness as graph program. Correctness in our sense means that all graph programs over the considered graph transformation system are correct. \par
\emph{Correctness} and \emph{verification} of programs are addressed in \cite{Olderog97} including also verification of \emph{non-deterministic} and \emph{distributed} programs. \par
In \cite{Flick16}, correctness of \emph{graph programs under adverse conditions} is considered. Unlike our interpretation of adverse conditions, the system is \emph{controlled} while the environment is \emph{uncontrolled}. Correctness is defined w.r.t. pre- and postconditions which are realized in form of $\mu$-\emph{conditions} which generalize graph conditions. \par
\emph{Graph grammars} under adverse conditions are considered in \cite{Peuser18} as well as \emph{temporal graph conditions}.  In contrast to our approach, system and environment \emph{play against each other}, i.e., their interaction is modeled by a \emph{game}. Correctness is achieved by a \emph{winning strategy}. \par
\ignore{In \cite{Manna69}, the relation between correctness of programs and first-order formulas is investigated. More precisely, verification is reduced to the \emph{satisfiability} of first-order formulas. \\} We adhered to \cite{Emerson90,Baier08} for the syntax and semantics of LTL\longv{/CTL}. Methods for \emph{model checking} are presented also in \cite{Baier08}. Tool support for LTL\longv{/CTL} model checking is due to GROOVE \cite{GROOVE} available. \ignore{By contrast, the tool Henshin \cite{Henshin} supports the more expressive $\mu$-calculus.}
In \cite{Vardi07}, a concept for LTL on \emph{finite traces} is introduced. For our purposes, completing finite sequences to infinite by applying the identical rule, and using the satisfaction notion for infinite sequences, is sufficient. \par
In \cite{Giese19}, a \emph{metric temporal logic} over \emph{typed attributed graphs} is introduced. The atoms are graph conditions. Instead of \emph{discrete} temporality like in the case of LTL, the \emph{continuous} set of non-negative real numbers is considered. The main point of the definition is the until-operator w.r.t.\ an interval of real numbers. This approach also incorporates the tracking of elements over time, which is not possible with LTL/CTL constraints.

%\shortv{\newpage}
\section{Conclusion and Outlook}\label{conc}
To sum up, in this paper, we have provided \begin{itemize} \item[(1)] a definition of joint graph transformation systems which involve the system, the environment, and a regulation automaton modeling the interaction between them, \item[(2)] \longv{three}\shortv{two} instances of correctness notions for joint graph transformation systems, which generalize the correctness for graph transformation systems, \item[(3)] a reduction theorem stating that \shortv{both}\longv{the presented} correctness notions are expressible as LTL\longv{/CTL} constraints. \end{itemize}
This approach points out a way for verification under adverse conditions subject to the condition that LTL\longv{/CTL} model checking is possible in the considered case. This may not always be the case since we consider, in general, infinite state spaces (infinite sets of graphs).

\ignore{
\begin{center}\scalebox{0.8}{
\tikzset{
    state/.style={
           rectangle,
           rounded corners,
           draw=black,
           minimum height=2em,
           inner sep=2pt,
           text centered,
           },
}

\begin{tikzpicture}[->,>=stealth']
\node[state,fill=whitegray] (box1) at (0,0) { \begin{tabular}{c} \\ \\ modeling \\ adverse \\ conditions \\ \hspace*{1mm} \\ \hspace*{1mm} \end{tabular}};
\node[state,fill=whitegray] (box1) at (7,0) { \begin{tabular}{c}\\  \\ model \\ checking \\ algorithm \\ \hspace*{1mm} \\ \hspace*{1mm} \end{tabular}};

\draw (-5,1.2) -- node[above]{system}(-0.9,1.2);  
\draw (-5,0.4) -- node[above]{environment}(-0.9,0.4);
\draw (-5,-0.4) -- node[above]{regulation automaton}(-0.9,-0.4); 
\draw (-5,-1.2) -- node[above]{pre- and postcondition}(-0.9,-1.2); 
\draw (0.9,0.4) -- node[above]{joint GTS}(6.15,0.4); 
\draw (0.9,-0.4) -- node[above]{LTL\longv{/CTL} constraint}(6.15,-0.4); 
\draw (7.85,0.4) -- node[above]{yes}(10.5,0.4); 
\draw (7.85,-0.4) -- node[above]{no}(10.5,-0.4); 
\end{tikzpicture}
}
\end{center}}

\begin{center}\scalebox{0.75}{
\tikzset{
    state/.style={
           rectangle,
           rounded corners,
           draw=black,
           minimum height=2em,
           inner sep=2pt,
           text centered,
           },
      rec/.style={
          rectangle,
           rounded corners,
           dashed,
           draw=black,
           minimum height=2em,
           inner sep=2pt,
           text centered,
           },
}

\begin{tikzpicture}[->,>=stealth']
\node[state,fill=whitegray] (box1) at (0,0.5) { \begin{tabular}{c} \\ \\ \hspace*{3mm}construct\hspace*{3mm} \\ \hspace*{1mm} \\ \hspace*{1mm} \end{tabular}};
\node[state,fill=whitegray] (box3) at (0,-1.2) { \begin{tabular}{c}\hspace*{2mm} reduction \hspace*{2mm}\end{tabular}};
\node[state,fill=whitegray] (box1) at (7,-0.4) { \begin{tabular}{c}  \\ model \\ checking \\ algorithm \\ \hspace*{1mm}  \end{tabular}};

\draw (-5.2,1.2) -- node[above]{system}(-1.3,1.2);  
\draw (-5.2,0.4) -- node[above]{environment}((-1.3,0.4);
\draw (-5.2,-0.4) -- node[above]{regulation automaton}(-1.3,-0.4); 
\draw (-5.2,-1.2) -- node[above]{correctness notion}(-1.3,-1.2); 
\draw (1.3,0.4) -- node[above]{joint GTS}(5.95,0.4); 
\draw (1.3,-1.2) -- node[above]{LTL\longv{/CTL} constraint}(5.95,-1.2); 
\draw (8.05,0.4) -- node[above]{\emph{yes}}(10.5,0.4); 
\draw (8.05,-1.2) -- node[above]{\emph{no}}(10.5,-1.2); 

\node[rec] (box4) at (0.3,0) { \begin{tabular}{c} \\ \\ \\\\ \hspace*{9.9cm} \\  \hspace*{1mm}\\  \hspace*{1mm}\\  \hspace*{1mm}\end{tabular} };
\node (box5) at (0.3,2.2) {modeling adverse conditions};
\end{tikzpicture}
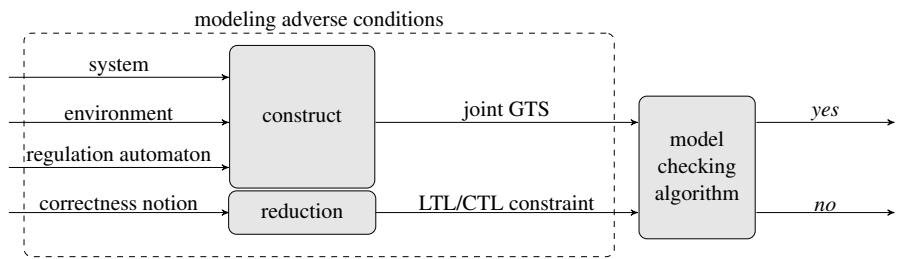
} \end{center}\begin{center} Fig. 10: overview of the presented construction and reduction \end{center}
Our future work will focus mainly on verification under adverse conditions using available model checking techniques. More precisely, we consider the following problems: \begin{itemize} \item[(1)] Model checking for joint graph transformation systems (also using GROOVE \cite{GROOVE}). \item[(2)] \emph{Well-structuredness} \cite{Koenig14} of joint graph transformation systems. \item[(3)] Integration of \emph{probabilistic} graph transformation \cite{Heckel04}. \end{itemize} Regarding (2), we conjecture that the class of well-structured graph transformation systems (w.r.t.\ the same well-quasi order) is ``closed'' under the construction of joint graph transformation systems. \vspace*{4mm} \\
\textbf{Acknowledgement.} We are grateful to Annegret Habel, Ernst-R\"udiger Olderog, Christian Sandmann, Nick W\"urdemann, and the anonymous reviewers for their helpful comments to this paper.

\begingroup

\bibliographystyle{eptcs}
\bibliography{ProposalBibAll}
\apx{\input{ModelingAdverseConditionsPaper_PreFinal_append}}
\endgroup 
\end{document}